\documentclass{article}
\usepackage[utf8]{inputenc}
\usepackage{lscape}
\usepackage[margin=1.0in]{geometry} % showframe
\usepackage{graphicx}
\usepackage{natbib}
\usepackage[labelfont=bf]{caption}

\usepackage[capposition=top]{floatrow}
\usepackage[T1]{fontenc}
\usepackage{tikz}
\usepackage{authblk}
\usepackage{xcolor}
\usepackage{booktabs}
\usepackage{multirow}

\usepackage{amsmath,amsfonts,amssymb,bm,xcolor}
\usepackage{stmaryrd}

\usepackage{booktabs} 
\usepackage{nicefrac}  
\usepackage{microtype}
\usepackage[utf8]{inputenc}

\usepackage{wrapfig}
\usepackage{multirow}

\usepackage{amssymb}
\usepackage{amsthm}
\usepackage{bbm}
\usepackage{bbm}
\usepackage{array}
\usepackage[colorlinks=true,urlcolor=purple,
linkcolor=purple,citecolor=purple]{hyperref}
\usepackage{amsmath}
\usepackage{mathtools}
\usepackage{float}

\usepackage{lscape}
\usepackage{graphicx}

\usepackage{booktabs}
\usepackage{longtable}
\usepackage{threeparttable}
\usepackage{threeparttablex}

\usepackage{mathrsfs}
\usepackage{graphicx} 
\usepackage{color}
\usepackage{natbib}

\usepackage{algorithm}
\usepackage{algpseudocode}
\usepackage{thmtools,thm-restate} %%for restating results 

\makeatletter
\def\paragraph{\@startsection{paragraph}{4}%
  \z@\z@{-\fontdimen2\font}%
  {\normalfont\bfseries}}
\makeatother

\usepackage{tikz}
\usetikzlibrary{decorations.markings, backgrounds, calc, positioning, shapes, shadows, arrows, fit, automata}
\usepackage{tikz-cd}
\usetikzlibrary{arrows.meta}
\tikzset{>={Latex}}

\renewcommand\vec{\bm}
\newcommand{\mat}[1]{{\mathbf{#1}}}

\newcommand{\CD}{\mathop{{}D}}
\newcommand{\CDnok}{\mathop{{}D^{\mathrm{nk}}}}
\newcommand{\Din}[1]{{\mathcal{D}_{\mathrm{in}}(#1)}}
\newcommand{\degin}[1]{{d_{\mathrm{in}}(#1)}}
\newcommand{\degout}[1]{{d_{\mathrm{out}}(#1)}}
\newcommand{\Dout}[1]{{\mathcal{D}_{\mathrm{out}}(#1)}}
\newcommand{\Kcd}[1]{{\mathcal{K}(#1)}}
\newcommand{\Nbd}[1]{{\mathcal{N}^{\mathrm{CD}}(#1)}}
\newcommand{\Nbdnok}[1]{{\mathcal{N}^{\text{CD}_{\mathrm{nk}}}(#1)}}

\newcommand{\Cnet}{\mathop{{}\mathcal{G}}}
\newcommand{\BC}{\mathop{{}B}}
\newcommand{\IDC}{\mathop{{}Q}}
\newcommand{\PR}{\mathop{{}\Pi}}
\newcommand{\BCnorm}{\mathop{{}B^{\mathrm{CD}}}}
\newcommand{\BCnormnok}{\mathop{{}B^{\mathrm{nk}}}}
\newcommand{\pnorm}{\mathop{{}p^{\mathrm{CD}}}}
\newcommand{\pnormnok}{\mathop{{}p^{\mathrm{nk}}}}

\newcommand{\shortestp}{\sigma}

\newcounter{desccount}

\newcommand{\descref}[1]{\hyperref[#1]{#1}}

\newtheorem{theorem}{Theorem}

\newtheorem{proposition}[theorem]{Proposition}

\theoremstyle{definition}

\newtheorem{definition}{Definition}

\makeatletter
\newcommand{\pushright}[1]{\ifmeasuring@#1\else\omit\hfill$\displaystyle#1$\fi\ignorespaces}
\newcommand{\pushleft}[1]{\ifmeasuring@#1\else\omit$\displaystyle#1$\hfill\fi\ignorespaces}
\makeatother

\usepackage{titlesec}
\titleformat*{\section}{\normalsize\bfseries}
\titleformat*{\subsection}{\normalsize\bfseries}
\titleformat*{\subsubsection}{\normalsize\bfseries}
\allowdisplaybreaks

\title{A Mathematical Framework for Citation Disruption\thanks{We thank the National Science Foundation for financial support of work related to this project (grants 1829168, 1932596, and 1829302).}}
\author[1]{Thomas Gebhart}
\author[2]{Russell Funk}
\affil[1]{{\small{Computer Science and Engineering, University of Minnesota}}}
\affil[2]{{\small{Carlson School of Management, University of Minnesota}}}
\date{}

\begin{document}

\maketitle

\begin{abstract}
Many theories of scientific and technological progress imagine science as an iterative, developmental process periodically interrupted by innovations which disrupt and restructure the status quo. Due to the immense societal value created by these disruptive scientific and technological innovations, accurately operationalizing this perspective into quantifiable terms represents a key challenge for researchers seeking to understand the history and mechanisms underlying scientific and technological progress. Researchers in this area have recently proposed a number of quantitative measures that seek to quantify the extent to which works in science and technology are disruptive with respect to their scientific context. While these disruption measures show promise in their ability to quantify potentially disruptive works of science and technology, their definitions are bespoke to the science of science and lack a broader theoretical framework, obscuring their interrelationships and limiting their adoption within broader network science paradigms. We propose a mathematical framework for conceptualizing and measuring disruptive scientific contributions within citation networks through the lens of network centrality, and formally relate the CD Index disruption measure and its variants to betweenness centrality. By reinterpreting disruption through the lens of centrality, we unify a number of existing citation-based disruption measures while simultaneously providing natural generalizations which enjoy empirical and computational efficiencies. We validate these theoretical observations by computing a variety of disruption measures on real citation data and find that computing these centrality-based disruption measures over ego networks of increasing radius results in better discernment of future award-winning scientific innovations relative to conventional disruption metrics which rely on local citation context alone. This work extends the theoretical foundations and potential applications of citation disruption measures and clarifies the relationship to other notions of scholarly importance, highlighting fruitful connections between bibliometrics and network science.
\end{abstract}

\clearpage

\section{Introduction}

Scientific and technological knowledge is characterized by its dynamic nature, constantly evolving through the contributions of scientists and inventors~\citep{popper2005logic, mokyr1992lever, arthur2009nature, fleck2012genesis, arthur2007structure, mokyr1992lever}. 
This evolution is driven by a combination of developmental improvements and disruptive breakthroughs, which shape the trajectory of progress in science and technology. 
The Kuhnian view of scientific and technological progress imagines science as an iterative process, developing incrementally through time, periodically interrupted by periods of revolution, wherein major paradigm shifts disrupt the accepted principles beheld by the preceding ``normal'' regime~\citep{kuhn1962structure}. 
In a similar vein, the ``creative destruction'' theory of economic innovation, popularized by Schumpeter, posits that industrial progress is driven by the incessant destruction of old technologies by the new~\citep{schumpeter1942capitalism}. 
From general relativity to penicillin, DNA to the internet, artifacts of this revolutionary potential of science continuously restructure society and our shared understanding of the universe. 
Due to the immense societal value created by these disruptive scientific and technological innovations, accurately operationalizing this perspective into quantifiable terms represents a key challenge for researchers seeking to understand the history and mechanisms underlying scientific and technological progress~\citep{fortunato2018science}. 

Recently, a number of promising network-theoretic measurements of such scientific and technological disruption have emerged towards this goal~\citep{funk2017dynamic, bornmann2020disruptive, leydesdorff2021proposal}.
Buoyed by the advent of massive, electronic bibliometric datasets, these \emph{disruption measures} operationalize the revolutionary interpretation of scientific progress by evaluating the extent to which particular works of science or technology restructure their local knowledge niche, as defined by their relationship to their neighborhood within a citation network. 
These disruption measures have shown promise in their ability to pick out scientific and technological works that are interpreted as paradigm-shifting~\citep{bornmann2020disruption} while remaining distinct from citation count, a widely-acknowledge--but sometimes flawed~\citep{bornmann2008citation}--indicator of innovative value.
These disruption measures have received wide adoption within the field of science and innovation studies, and have begun to appear as dependent variables in a number of metascientific analyses measuring the differences in scientific achievement with respect to team size~\citep{wu2019large}, the effects of topical disagreements on scientific output~\citep{lin2022new}, and the observed slowing pace of scientific disruption altogether~\citep{park2023papers}.

Despite this empirical success, these disruption measures are largely lacking in a robust mathematical foundation. 
Their definitions typically rely on counts of papers within bespoke constructions of network neighborhoods and are heavily dependent on their citation network context. 
This lack of mathematical formalism hides the relationships among competing measures of citation disruption, limits the wider application of these disruption measures to other non-bibliometric network-theoretic domains, hinders the development of more extensive models of scientific innovation, and obscures their position within the broader network social science paradigm~\cite{borgatti2009network}.

% The earliest presentation of this class of disruption measures, the CD Index~\citep{funk2017dynamic}, is motivated as a measure of the extent to which attention is shifted away from particular technologies following the introduction of a new technology. 
% While this motivation broadly aligns with the Kuhnian and Schumpeterian perspectives discussed above, the implementation of this measure is specific to the particular context of a patent and its immediate successor and predecessor patents with respect to a larger citation network.
% This definitional specificity is a detriment to broader applicability of these measures. 
% It both masks the measurement assumptions underlying the index and obfuscates its relevance to other measurement contexts outside of bibliometric data. 
% In addition, the majority of the disruption measures proposed since the CD Index have been motivated as either improvements or extensions to the original CD Index implementation and therefore inherit the same definitional specificity as their predecessors, resulting in a menagerie of disruption measures whose relationships to each other and other network-theoretic concepts are muddled by context dependence. 

In this work, we bridge this theoretical gap by providing a mathematical framework for the definition of citation disruption via network centrality.
Specifically, we re-conceptualize a popular measure of citation disruption, the CD Index~\citep{funk2017dynamic}, as a measure of betweenness centrality, a well-studied concept in network science that seeks to measure the ``importance'' of nodes within a network as a function of the proportion of shortest paths passing through each node~\citep{anthonisse1971rush, freeman1977set}. 
In addition, we show that this centrality framework for measuring scientific and technological disruption is both flexible enough to express many of the objectives sought by a citation-based disruption measure while also recovering existing disruption measures like citation count and variants of the CD Index~\citep{funk2017dynamic, leydesdorff2021proposal, bu2021multidimensional, leibel2023we} as special cases. 
The relationship between citation disruption and other frequently-used measures of scholarly importance becomes clear under this network-theoretic reframing, thereby broadening the relevance of disruption measurement to other network science domains and vice versa.

In addition to unifying a number of existing disruption measures, this centrality definition of disruption also points towards natural extensions to existing disruption measures which are better aligned with their theoretical motivations and more robust to the noisiness of real-world citation patterns. 
We verify the empirical potential of these extensions by observing that they are more discerning of award-winning scientific and technological innovations compared to other disruption measures, like the CD Index, or citation count.

\section{Measuring Disruption}\label{sec:disruption}

Foundational theories of scientific and technological change highlight the existence of two types of breakthroughs~\citep{kuhn1962structure, schumpeter1942capitalism, dosi1982technological}. The first type consists of contributions that enhance and refine existing streams of knowledge, thereby consolidating the status quo \citep{enos1958measure, david1990dynamo, rosenberg1982inside}. These developmental improvements build upon established theories and methodologies, refining them for greater accuracy, efficiency, or applicability, thereby making them more valuable \citep{enos1962invention}. The second type of breakthroughs challenge and disrupt existing knowledge, rendering it obsolete and propelling science and technology in new and unforeseen directions \citep{tushman1986technological}. These breakthroughs have the potential to revolutionize entire fields, opening up new avenues of inquiry and application. By embracing both types of breakthroughs, the scientific and technological community continually pushes the boundaries of what is known and reshapes our understanding of the world, paving the way for transformative advancements and discoveries.

\subsection{Citation Networks}

Given the abstract and multifaceted nature of scientific and technological knowledge, precisely measuring and quantifying the distinction between developmental and disruptive intellectual contributions poses a significant challenge. However, large-scale bibliometric data, particularly in the form of published scientific papers and patented technologies, offer a valuable context within which to begin making such quantifications \citep{price1963, jaffe2002patents}. The vast body of scientific literature and patent records provides a wealth of information that enables researchers to analyze and trace the evolution of ideas, concepts, and technologies \citep{liu2023data, wang2021science}. Papers and patents not only present novel ideas but also make citations to prior works, thereby establishing a conceptual genealogy. Analysis of the evolution of citation networks therefore enables tracing of the influence and impact of specific contributions, discernment of patterns of continuity and transformation, and consequently, one approach for the identification of disruptive breakthroughs. While acknowledging the inherent complexities \citep{bornmann2015methods,tahamtan2019citation, tahamtan2018core,bornmann2020can, waltman2016review}, leveraging bibliometric data in the study of scientific and technological evolution provides valuable insights into the dynamics of knowledge advancement and facilitates a more nuanced understanding of the disruptive forces driving innovation \citep{wu2019large, figueiredo2019quantifying, andrade2020measuring, azoulay2019does, leahey2023types, zeng2021fresh, chu2021slowed, wang2021science, wang2023quantifying}.

\begin{wrapfigure}{r}{0.5\textwidth}
    \begin{center}
    \includegraphics[width=\textwidth]{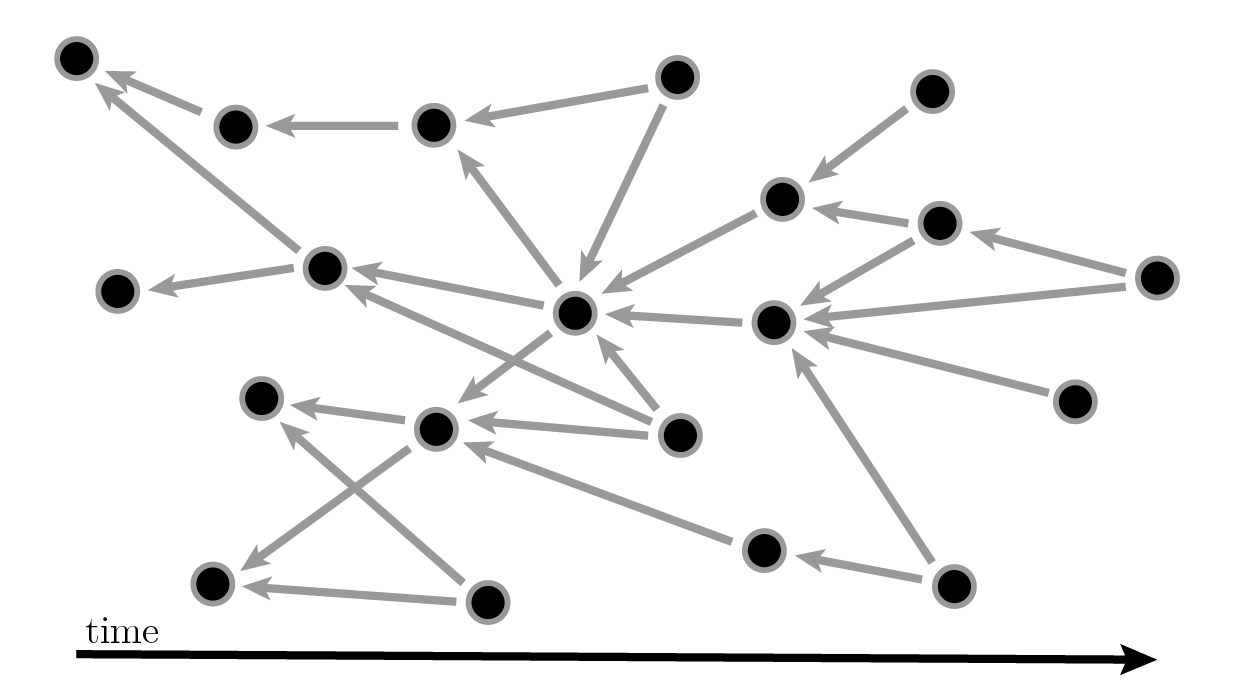}
    \end{center}
    \caption{An example citation network $\Cnet$.}
    \vspace{-20pt}
    \label{fig:citation_network}
\end{wrapfigure}

We define these citation networks as follows.
\begin{definition}
    Given a collection of papers $V$, a \emph{citation network} $\Cnet = (V,E)$ is a directed, unweighted graph formed by directed edges $(u,v) \in E$ connecting $u \in V$ directed towards $v \in V$ if paper $u$ cites paper $v$. 
\end{definition}
Here we have used ``papers'' as shorthand to refer to any attributed scientific or technological work which may be situated within a citation network (academic publications, books, patents, etc.). 
We will continue to use this nomenclature in the following sections though the results extend to any corpus of work which engages in attribution. 
Under some minor assumptions, we can view $\Cnet$ as acyclic, with the directedness of edges describing an implicit temporal ordering of papers, such that the existence of edge $(u,v)$ implies paper $v$ was published before paper $u$.
It will sometimes be convenient to overload notation and write $V(\Cnet)$ or $E(\Cnet)$ to refer to the set of nodes or set of edges, respectively, of graph $\Cnet$. 

Note that each paper $v$ induces a (possibly trivial) subgraph of papers $\Dout{v}$ composed of edges recording $v$'s citations to prior work $E(\Dout{v}) = \{(v,u) \mid u \in V(\Dout{v}) \}$. 
This $\Dout{v}$ subgraph represents paper $v$'s approximation of how the concepts and ideas presented in $v$ relate to or are otherwise inspired by the cited collection of prior work $V(\Dout{v}) \setminus \{v\}$. 
Conversely, $v$ also induces a (possibly trivial) subgraph $\Din{v}$ composed of edges $E(\Din{v}) = \{(u,v) \mid u \in V(\Din{v}) \}$ which connect each paper citing $v$ to $v$.
The set $V(\Din{v}) \setminus \{v\}$ may be interpreted as the collection of papers that were directly impacted by or otherwise derived ideas from paper $v$ specifically. 

\subsection{Properties of Disruption Measures}

The definition and measurement of disruption within citation networks necessitate an examination of how intellectual contributions alter the value of prior streams of knowledge upon which they build~\citep{funk2017dynamic, park2023papers}. At the core of the notion of disruption lies the transformative effect that contributions have on these streams, simultaneously propelling them in new directions while breaking with the past, resulting in a decrease in the use of preceding works. Conversely, developmental contributions enhance the value and utility of previous work, increasing its usage. Therefore, a quantitative measure of disruption should primarily focus on characterizing whether and how a paper alters the use of its predecessors. Within the context of citation networks, this can be accomplished by evaluating the degree by which future works cite the prior works referenced by a focal paper. 

Beyond this fundamental requirement, we further suggest that such a measure should account for the intricate interconnectedness of scientific and technological knowledge, acknowledging the potential for both direct and indirect influences of a particular work. Specifically, it should be capable of characterizing neighborhoods of influence of varying sizes, capturing the nuanced cascading effects on subsequent scientific and technological development. Further, while theories of scientific and technological change often discuss disruption in categorical or binary terms, it is more appropriate to consider disruptiveness as a measure of degree~\citep{funk2017dynamic}. Some works fully eclipse the prior streams of work upon which they build, while others cast more partial shadows. Therefore, an ideal measure should exhibit a continuous nature, allowing for the quantification of these gradations.

In summary, given a citation network $\Cnet = (V,E)$, we seek to derive a \emph{disruption measure} $\Delta: V \to \mathbb{R}$ which captures the extent to which paper $v$ is disruptive with respect to the rest of $\Cnet$. 
Specifically, we seek a function $\Delta$, dependent on $\Cnet$, which
\begin{enumerate}
\item\label{enum:temporal_ordering} respects the temporal ordering of $\Cnet$,
\item\label{enum:citation_degree} measures the degree by which future works cite the prior works referenced by a focal paper,
\item\label{enum:indirect_influence} is sensitive to direct and indirect influence on future works,
\item\label{enum:continuity} is continuous with respect to the disruptive effects measured.
\end{enumerate}

We will refer back to these properties of disruption measures in the next section when we introduce network centrality, observing that many centrality measures happen to satisfy these requirements.

\section{Measuring Disruption with Network Centrality}\label{sec:centrality}

While there are no precise boundaries for its definition within network science, we define \emph{centrality} as a class of functions defined on networks which measure the structural or informational ``importance'' of nodes within the network.
Network centrality has a storied history within the social sciences, with the earliest application of this concept, closeness centrality, appearing at least as early as 1950 as an inverse measure of average distance to each node in the graph for use in evaluating communication efficiency in problem solving across different social group topologies~\citep{bavelas1950communication}. 
Since then, the number and variety of centrality measures have grown substantially~\citep{newman2018networks}. 
This growth is due to the fact that the notion of importance is highly context-dependent: importance in a social network may differ from that of a biophysical network, and different still from that a transportation network.\footnote{See \citet{landherr2010critical} for a review of centrality in social networks, \citet{ghasemi2014centrality} for usage of centrality in biological networks, and  and Chapter 7~\citet{newman2018networks} for a general introduction to some well-known centrality measures. \citet{bloch2023centrality} propose taxonomy of centrality measures.}

In this section, we provide an introduction to some well-known centrality measures and observe that all of the desired properties of a disruption measure discussed in Section~\ref{sec:disruption} can be satisfied by both betweenness and Pagerank centrality. 

\subsection{Degree Centrality}

\begin{wrapfigure}{r}{0.5\textwidth}
    \begin{center}
    \includegraphics[width=\textwidth]{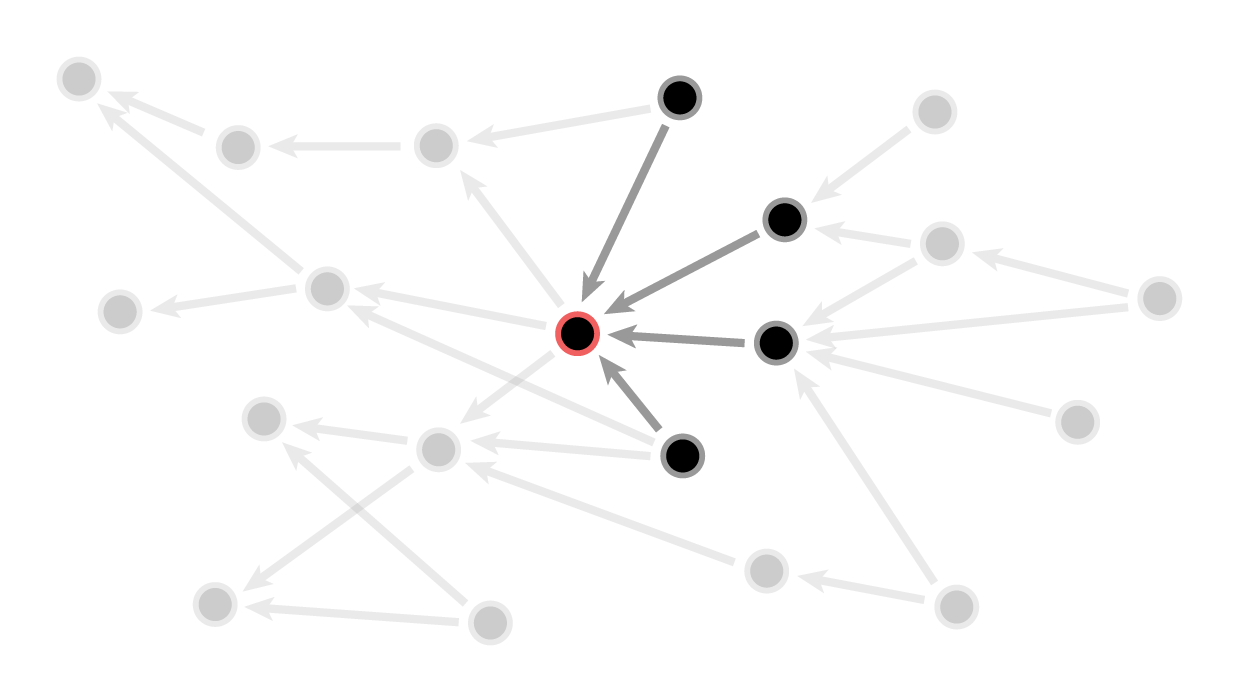}
    \end{center}
    \caption{The in-degree subgraph $\Din{v}$ of a citation network $\Cnet$.}
    \vspace{-20pt}
    \label{fig:indegree}
\end{wrapfigure}

The degree centrality of a node measures its number of incident edges. 
Over directed graphs, degree centrality subdivides into out-degree and in-degree centrality, depending on the orientation of the nodes incident edges. 
The latter of these centrality measures is relevant for our uses, so we provide a proper definition.
\begin{definition}
    Given a directed graph $\mathcal{H} = (V,E)$, define the in-neighborhood $\Din{v}$ of node $v$ with node set $V(\Din{v}) = \{v\} \cup \{u \mid (u,v) \in E\}$ and edge set $E(\Din{v}) = \{(u,v) \mid u \in V(\Din{v}) \}$. The \emph{in-degree centrality} $\IDC: V \to \mathbb{Z}_{\geq 0}$ of node $v$ is given by $\IDC(v) = |V(\Din{v})|-1 = \degin{v}$.
\end{definition} 
Evaluated over a citation network $\Cnet$, in-degree centrality satisfies disruption Properties~\ref{enum:temporal_ordering} and ~\ref{enum:continuity}, as the directedness of $\Cnet$ reflects the temporal publication order of papers and $\degin{v}$ may be arbitrarily large, respectively. 
However, $\IDC$ does not satisfy Properties~\ref{enum:citation_degree} and~\ref{enum:indirect_influence} due to its focus only on the citing works of focal node $v$. 

\subsection{Betweenness Centrality}

Betweenness centrality~\citep{anthonisse1971rush, freeman1977set} measures the importance of a node $v$ in a network by counting the proportion of shortest paths between each pair of vertices in the graph which pass through $v$. 
\begin{definition}\label{def:betweenness_centrality}
    The \emph{betweenness centrality} $\BC: V \to \mathbb{R}_{\geq 0}$ of a node $v$ within a graph $\mathcal{H} = (V,E)$ is given by
    \begin{equation*}
        \BC(v) = \frac{1}{p}\sum\limits_{s \in V \setminus \{v\}} \sum\limits_{t \in V \setminus \{v\}} \frac{\shortestp(s,t \mid v)}{\shortestp(s,t)}
    \end{equation*}
    where $\shortestp(s,t)$ is the number of shortest paths between nodes $s$ and $t$ in $\mathcal{H}$, $\shortestp(s,t \mid v)$ is the number of shortest paths originating at node $s$ and terminating at $t$ which pass through $v$, and $p$ is a normalization constant. 
\end{definition}
Evaluated over citation network $\Cnet$, betweenness centrality satisfies all of the disruption Properties listed in Section~\ref{sec:disruption}. 
Because $\BC(v)$ is a relative count of shortest paths from arbitrary nodes across a temporally-ordered citation network, if $v$ acts as a bottleneck in the citation network, requiring its visitation along a shortest path between future and past works of paper $v$, its betweeenness centrality will be high. 
By contrast, if the cited works of $v$ are frequently cited by future works of $v$, its betweenness will be low, as there are multiple shortest paths from future to past work which route around $v$. 

Using each node's inclusion in shortest paths as the measurement of ``importance'' endows the betweenness centrality measure of disruption with particular semantics. 
This geodesic betweenness condition implies that if any path between nodes $s$ and $t$ is shorter than the shortest path between $s$ and $t$ passing through $v$, then $v$ does not inherit any ``importance'' with respect to those pairs of nodes. 
One can imagine relaxing these austere flow constraints such that if a the path(s) between $s$ and $t$ through $v$ are ``close'' to being important, then $v$ still inherits some centrality from this relationship. 
Pagerank centrality represents one such relaxation by replacing the shortest path betweenness measure with a visitation probability determined by a random walk. 

\subsection{Pagerank Centrality}

Pagerank centrality~\citep{page1998pagerank} reinterprets the directed graph $\mathcal{H} = (V,E)$ with adjacency matrix $\mat{A}$ as a Markov chain with transition probabilities $\mat{P} = \mat{D}^{-1}\mat{A}$, where $\mat{D} = \mathrm{diag}(\mat{A}\vec{1})$ is a diagonal matrix of node out-degrees and $\vec{1}$ is a vector of $1$'s.
Pagerank assigns centrality based on the stationary distribution of a random walk on this Markov chain. 
Dangling nodes in $\mathcal{H}$ which have zero out-degree eventually capture all probability mass, and thus trivialize the long-run random walk dynamics.
To combat this behavior, we connect these dangling nodes to other non-dangling nodes in the graph according to probability vector $\vec{\gamma}$, resulting in a new stochastic matrix $\bar{\mat{P}}$.

\begin{definition}
    Given directed graph $\mathcal{H} = (V,E)$ with stochastic transition matrix $\bar{\mat{P}}$ determined by personalization vector $\vec{\gamma} > 0$ and teleportation probability $\alpha$, the \emph{personalized Pagerank centrality} $\PR: V \to \mathbb{R}_{\geq 0}$ of node $v \in V$ is given by $\PR(v) = \vec{\pi}_{v}$ where $\vec{\pi}$ is the solution to the eigenvalue problem
    \begin{equation}\label{eq:pagerank}
        \vec{\pi}^\top (\alpha \bar{\mat{P}} + (1 - \alpha) \vec{1}\vec{\gamma}^\top) = \vec{\pi}^\top .
    \end{equation}
\end{definition}
When $\vec{\gamma} = \vec{1}|V|^{-1}$, we recover the original Pagerank algorithm which assigns equal teleportation probability between each pair of nodes in the network.  

Equation~\ref{eq:pagerank} solves for the stationary distribution of a random walk on $\mathcal{H}$ which teleports to new nodes with probability $1-\alpha$. 
Thus, measured over citation network $\Cnet$, we may interpret $\PR(v)$ as measuring the likelihood a random walker moving backwards through time along paper citations passes through node $v$ (Properties~\ref{enum:temporal_ordering},~\ref{enum:continuity}). 
If paper $v$ is highly-cited or is cited by a number of highly-cited papers, $\PR(v)$ will be high because $v$ has many opportunities to be visited along a random walk (Properties~\ref{enum:citation_degree},~\ref{enum:indirect_influence}). 
Thus, we may interpret $\PR$ as a disruption measure in the sense that papers with high Pagerank will be those which will be most likely traversed when walking the citation network between present and past works. 
Under this interpretation, we see Pagerank satisfies all four properties of a disruption measure given in Section~\ref{sec:disruption}. 

\section{Existing Disruption Measures are Centrality Measures}

We will now show that many of the measures already in use for quantifying scientific and technological disruption on citation networks may be rewritten as specific instances of the well-known centrality measures given in the previous section. 
In particular, we show that citation count and in-degree centrality are interchangeable, and that the CD Index is a shifted version of betweenness centrality evaluated over a bespoke neighborhood graph around each node in the network. 

\subsection{Citation Count}\label{sec:citation_count}
Citation count is a ubiquitous measure of scientific and technological impact that records the number of times an individual paper $v$ has been cited. 
Embedded in a citation network $\Cnet$, citation count and in-degree centrality $\IDC$ are equivalent.
This equivalence implies citation count fails to satisfy Properties~\ref{enum:citation_degree} and~\ref{enum:indirect_influence} given in Section~\ref{sec:disruption}.

\subsection{The CD Index}\label{sec:cd_index}

The CD Index~\citep{funk2017dynamic} is a citation-based measure $\CD(v)$ of the ``disruptive'' effect that a scientific work $v$ introduces with respect to its topic-specific context within a citation network. 
This topic context of $v$ is typically proxied by observing the citation patterns of a neighborhood $\Nbd{v}$ around $v$ within the broader citation graph. 
The CD Index, then, is a measure over the possible configurations of the citation neighborhood $\Nbd{v}$, assigning higher values to $v$ which have high ``importance'' to the connectivity of $\Nbd{v}$ and low values to those with relatively low ``importance'' within $\Nbd{v}$.

A number of distinct disruption measures have been introduced under the ``CD Index'' moniker~\citep{funk2017dynamic, bornmann2020disruption, wang2023quantifying, leydesdorff2021proposal, chen2021destabilization, li2022measuring, deng2023enhancing, wu2019solo}. 
We will narrow our focus to two closely related definitions of the CD Index, which we denote $\CD$ and $\CDnok$, introduced in \cite{funk2017dynamic} and \cite{bornmann2020disruption}, respectively. 

The definition of the CD Index relies on the construction of a bespoke neighborhood subgraph $\Nbd{v}$ around node $v$. 
This neighborhood subgraph forms the basis of the CD Index and is given by the following union of graphs:
\begin{definition}\label{def:cd_neighborhood}
Given an ambient citation graph $\Cnet$, the \emph{CD Index neighborhood} $\Nbd{v}$ \emph{of node $v$}  is defined by
\begin{align}
    \Nbd{v} &= \Din{v} \cup \Dout{v} \cup \left(\bigcup\limits_{u \in V(\Dout{v})} \Din{u}\right)\nonumber \\
    &= \Din{v} \cup \Dout{v} \cup \Kcd{v}\label{eq:cd_neighborhood}
\end{align}
\end{definition}
where $$\Kcd{v} = \bigcup\limits_{u \in V(\Dout{v})} \Din{u}$$ is the subgraph composed of the union of in-citations for each node in the out-citation subgraph of $v$. 
\begin{figure}
    \centering
    \includegraphics[width=0.49\textwidth]{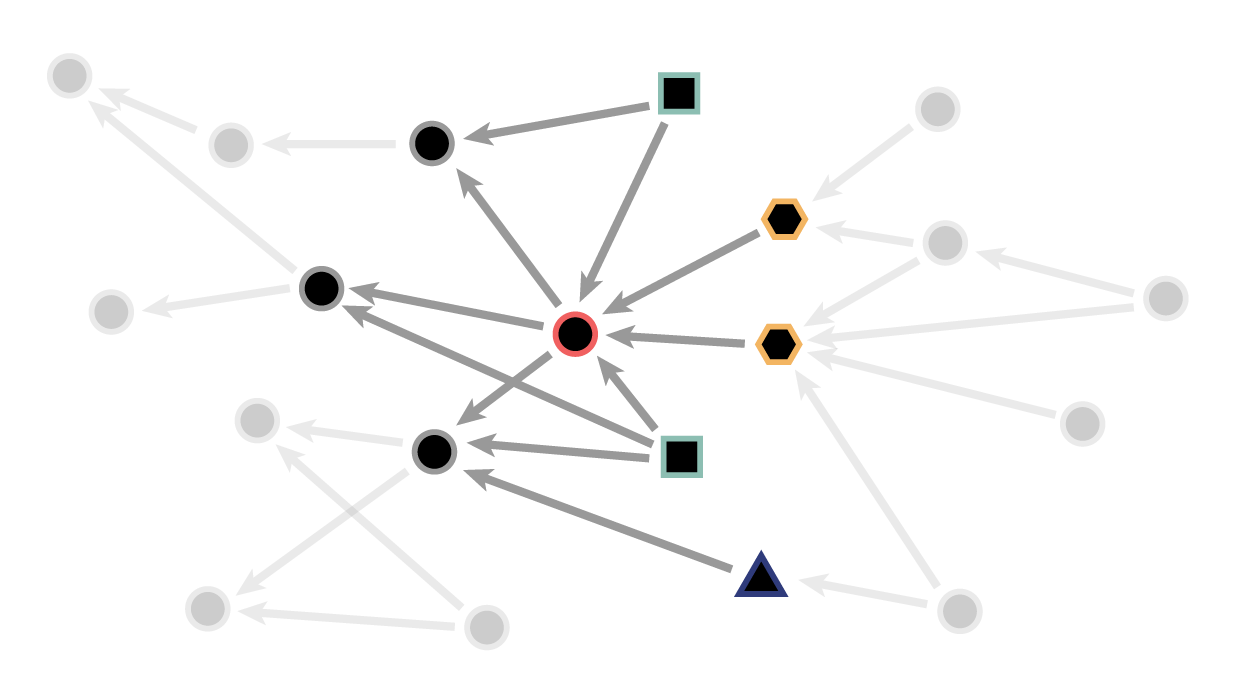}
    \includegraphics[width=0.49\textwidth]{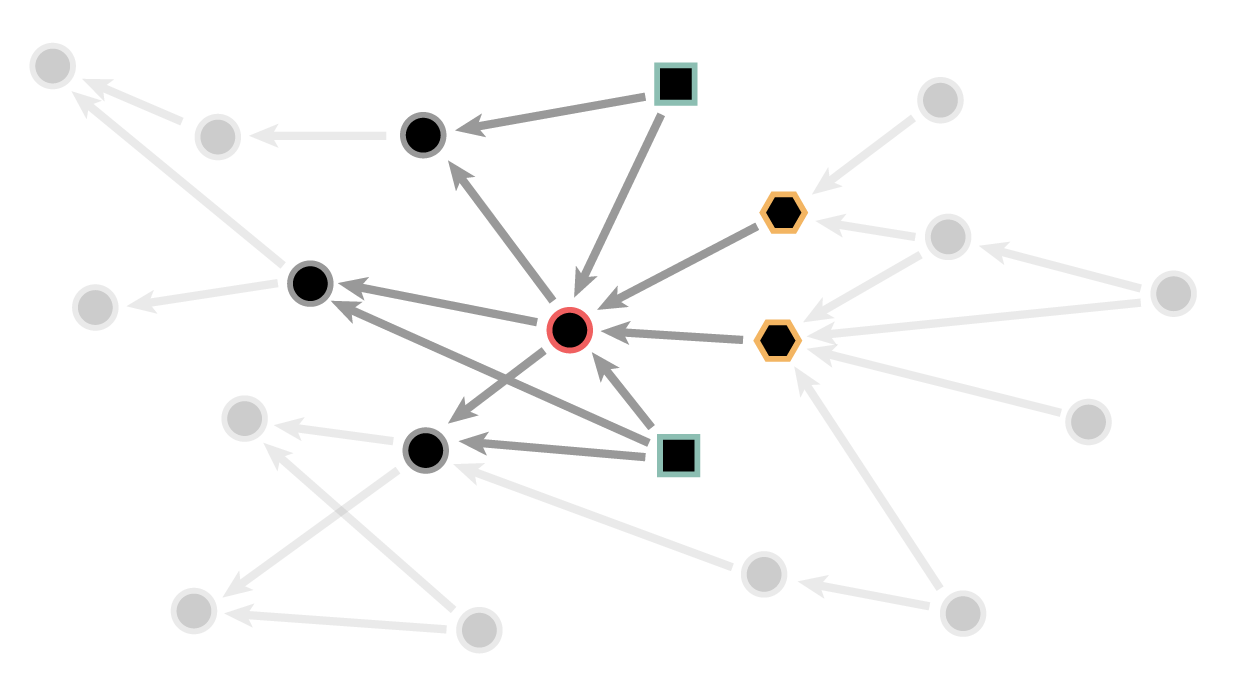}
    \caption{Left: the CD Index neighborhood subgraph $\Nbd{v}$. Right: the ``no-$k$'' CD Index neighborhood subgraph $\Nbdnok{v}$. The focal paper $v$ is denoted by a red border. $I$-type nodes are denoted by hexagons. $J$-type nodes are denoted by squares. $K$-type nodes are denoted by triangles.}
    \label{fig:neighborhood_subgraph}
\end{figure}
We can equivalently define Equation~\ref{eq:cd_neighborhood} in terms of its vertex and edge sets:
\begin{align*}
    V(\Nbd{v}) &= V(\Din{v}) \cup V(\Dout{v}) \cup V(\Kcd{v}) \\
    E(\Nbd{v}) &= E(\Din{v}) \cup E(\Dout{v}) \cup E(\Kcd{v}).
\end{align*}

Removing this $\Kcd{v}$ term from the graph union, we can define the the ``no-$k$'' neighborhood subgraph $\Nbdnok{v}$ centered at $v$ as the following vertex and edge sets:
\begin{align*}
     V(\Nbdnok{v}) &= V(\Din{v}) \cup V(\Dout{v}) \\
     E(\Nbdnok{v}) &= E(\Din{v}) \cup E(\Dout{v}) \cup \{(u,w) \mid u \in V(\Din{v}), w \in V(\Dout{v}) \}.
\end{align*}
In other words, $\Nbdnok{v}$ is composed of the union of the in- and out-subgraphs induced by $v$, but also includes the edges between these two subgraphs. 
The subgraph $\Nbdnok{v}$ forms the basis for the ``no-$k$'' CD Index $\CDnok$ formulation given in \cite{bornmann2020disruption}.
Note that the subgraphs $\Nbd{v}$ and $\Nbdnok{v}$ are both similar to the (1-hop) ego subgraph $\overline{\mathcal{N}}_1(v)$ (Figure~\ref{fig:k_hop_neighborhoods}), in that all of these graphs contain the union $\Din{v} \cup \Dout{v}$. 
However, $\Nbdnok{v}$ lacks any edges between nodes within $V(\Din{v})$, in addition to lacking edges between nodes within $V(\Dout{v})$.
The subgraph $\Nbd{v}$ also lacks these edges, and contains additional nodes from $\mathcal{K}(v)$. 
% Because $\Nbd{v}$ simplifies to $\Nbdnok{v}$ if one ignores other citations to the papers which $v$ cites, we will often refer only to $\Nbd{v}$ in the following work, making distinctions between the two neighborhood graph types only when necessary. 

With these neighborhood graph definitions, we can now provide definitions of the CD Index measures $\CD$ and $\CDnok$. 
Both of these measures are contingent on a labeling of the nodes in the $\Din{v}$ subgraph of $\Nbd{v}$. 
\begin{definition}\label{def:itype}
    Given a CD Index neighborhood subgraph $\Nbd{v}$ of node $v$, define the \emph{$I$-type nodes} of $\Nbd{v}$ as the set $I(v) = \{u \in V(\Din{v}) \mid V(\Dout{u}) \cup \{v\} \}$. 
\end{definition}
\noindent In other words, the $I$-type nodes of $\Nbd{v}$ are those which cite \emph{only} $v$ within $\Nbd{v}$. 
\begin{definition}\label{def:jtype}
    Given a CD Index neighborhood subgraph $\Nbd{v}$ of node $v$, define the \emph{$J$-type nodes} $\Nbd{v}$ as the set $J(v) = \{u \in V(\Din{v}) \mid \degout{u} > 1\} = \{u \in V(\Din{v}) \setminus I(v) \cup \{v\} \}$
\end{definition} 
Here $\degout{u} = |\Dout{u}|-1$ is the out-degree of node $u$.   
The papers in $J(v)$ are those which cite both $v$ and at least one of the papers that $v$ cites, forming the complement of $I(v)$ with respect to the set of papers citing $v$.
Note that prior two definitions apply also to $\Nbdnok{v}$, but the following does not.
\begin{definition}\label{def:ktype}
    Given a CD Index neighborhood subgraph $\Nbd{v}$ of node $v$, define the \emph{$K$-type nodes} of $\Nbd{v}$ as the set $K(v) = \{u \in V(\Kcd{v}) \mid u \not\in V(\Din{v})\} = V\left(\Nbd{v}\right) \setminus \left(V(\Din{v}) \cup V(\Dout{v})\right)$. 
\end{definition}
\noindent The set $K(v)$ accounts for the papers which cite the papers that $v$ cites, but do not cite $v$ directly.
Given the node labelings in Definitions~\ref{def:itype},~\ref{def:jtype},~\ref{def:ktype}, we can now define the CD Index measures $\CD$ and $\CDnok$ as ratios of the magnitudes of these sets.
\begin{definition}\label{def:cd_index}
    Given a neighborhood graph $\Nbd{v}$ derived from citation graph $\Cnet = (V,E)$ as defined in Definition~\ref{def:cd_neighborhood}, let $n_I(v) = |I(v)|$, $n_J(v) = |J(v)|$, and $n_K(v) = |K(v)|$ be the sizes of the three node sets defined above. The \emph{CD Index} $\CD: V \to [-1,1]$ is given by:
    \begin{align}\label{eq:cdindex}
        \CD(v) &= \frac{n_I(v) - n_J(v)}{n_I(v) + n_J(v) + n_K(v)} \\
        &= \frac{n_I(v) - n_J(v)}{n_V(v)-1}\nonumber
    \end{align}
\end{definition}
\noindent  where $n_V(v) = |V(\Nbd{v})|$ is the number of nodes in the neighborhood graph. 
\begin{definition}\label{def:cdindex_nok}
    Given a neighborhood graph $\Nbd{v}$ derived from citation graph $\Cnet$ as defined in Definition~\ref{def:cd_neighborhood}, let $n_I(v) = |I(v)|$ and $n_J(v) = |J(v)|$ be the sizes of the node sets defined above. The \emph{no-$k$ CD Index} $\CDnok: V(\Cnet) \to [-1,1]$ is given by:
    \begin{align}\label{eq:cdindex_nok}
        \CDnok(v) &= \frac{n_I(v) - n_J(v)}{n_I(v) + n_J(v)} \\
        &= \frac{n_I(v) - n_J(v)}{\degin{v}}\nonumber
    \end{align}
\end{definition}
where $\degin{v} = |\Din{v}|-1$ is the in-degree of node $v$.

Intuitively, $\CD$ measures the propensity of paper $v$ to attract citations to itself, drawing attention away from the work which came before it, relative to the total number of citations received by $v$ and the papers it cites. 
The no-$k$ CD Index $\CDnok$ measures a similar propensity, but focuses only on citations made among nodes that cite or are cited by $v$. 
Note that $\CD$ and $\CDnok$ satisfy Properties~\ref{enum:temporal_ordering},~\ref{enum:continuity}, and~\ref{enum:citation_degree} but fail to satisfy Property~\ref{enum:indirect_influence} due to the one-hop definition of $\Nbd{v}$. 
We will return to this observation in Section~\ref{sec:generalizing_disruption}.

A paper $v^*$ which maximizes the CD Index such that $\CD(v^*) = 1$ is one which has no $J$- or $K$-type citations: all subsequent works either cite $v^*$ or ignore $v^*$ and its cited work entirely. 
Such a $v^*$ would be a \emph{maximally-disruptive} paper according to the CD Index.
By contrast, a paper $v^-$ which minimizes the CD Index is one with only $J$-type citations: all of its subsequent citations cite both itself and at least one of its cited papers. 
Such a $v^-$ would be a \emph{maximally-consolidating} paper according to the CD Index.
Note that $v^*$ and $v^-$ also maximize and minimize $\CDnok$.

The CD  Index $\CD(v)$ is inherently measuring the importance of node $v$ with respect to the citation neighborhood $\Nbd{v}$ by defining a ratio of node types among the papers citing $v$. 
As discussed in Section~\ref{sec:centrality}, network centrality also provides a framework for measuring importance of nodes embedded within networks. 
We will now show that this relationship between the importance measured by the CD Index and the importance measured by betweenness centrality are equivalent up to a translation in their ranges.  

\subsection{CD Index as Betweenness Centrality}\label{sec:cd_as_betweenness}

Recall that we are interested in the importance of nodes over a citation graph $\Cnet$, which introduces a number of structural constraints.
In particular, we wish to compute the centrality of a node $v$ with respect to $\Nbd{v}$, which has a very particular structure and leads to the following proposition:
\begin{proposition}\label{prop:betweenness_simplification}
    Evaluated over the CD Index neighborhood graph $\Nbd{v}$ as described in Definition~\ref{def:cd_neighborhood}, the betweenness centrality $\BC(v)$ of node $v$ may be computed as the (normalized) count of paths passing through $v$:
    \begin{equation*}
        \BC(v) = \frac{1}{p}\sum\limits_{s \in V \setminus \{v\}} \sum\limits_{t \in V \setminus \{v\}} \shortestp(s,t \mid v)
    \end{equation*}
\end{proposition}
\begin{proof} 
By definition of $\Nbd{v}$, there is never more than one shortest path between $s$ and $t$ whenever one exists which passes through $v$, thereby making the denominator of Definition~\ref{def:betweenness_centrality} $\sigma(s,t)=1$ whenever $\sigma(s,t \mid v) \not= 0$. 
\end{proof}

Proposition~\ref{prop:betweenness_simplification} implies that we must only concern ourselves with counting the paths passing through $v$ in order to calculate betweenness centrality over $\Nbd{v}$. 
As an immediate corollary, it is easy to see that the vertex $v^*$, which induces subgraph $\Nbd{v^*}$ maximizing $\BC$, also maximizes $\CD$: papers citing $v^*$ only cite $v^*$ and not its cited papers. 

The normalization factor $p$ in the definition of betweenness is typically chosen in such a way as to make $\BC(v)$ comparable across all possible choices of underlying graph. 
For example, on an arbitrary directed network with $n$ nodes, choosing $p = (n-1)(n-2)$ accounts for all possible choices of directed edges between all pairs of nodes excluding the measured node $v$.
However, many of the possible edges enumerated by this combinatorial choice of $p$ are not realizable within $\Nbd{v}$, so a tighter normalization constant is available. 
\begin{proposition}\label{prop:betweenness_normalization}
    Evaluated over the CD Index neighborhood graph $\Nbd{v}$ as described in Definition~\ref{def:cd_neighborhood}, the normalization constants
    \begin{align*}
        \pnormnok(v) &= \degin{v} \degout{v} \\
        &= \degout{v}(n_I(v) + n_J(v))
    \end{align*}
    and 
    \begin{align*}
        \pnorm(v) &= \degout{v} \cdot |V(\Nbd{v}) \setminus \left(V(\Din{v}) \cup V(\Dout{v})\right)| \\
        % &= \degin{v} \cdot \degout{v} + \degout{v} \cdot n_K(v) \\ 
        &= \degout{v}(\degin{v} + n_K(v)) \\
        &= \degout{v}(n_I(v) + n_J(v) + n_K(v)) \\
        % &= \degout{v}(n_V(v)-\degout{v}-1)
    \end{align*}
    both normalize $\BC(v)$ to the range $[0,1]$. We denote the betweenness centrality normalized by $\pnorm$ by $\BCnorm$ and the beweenness centrality normalized by $\pnormnok$ by $\BCnormnok$. 
\end{proposition}
\begin{proof}
    This is a multiplicative normalization, so the lower bound of $\BC$ is unchanged at 0 when normalizing. 
    Thus, it suffices to show that the $v^*$ which induces the subgraph maximizing $\BCnorm$ is equal to $1$ under this normalization scheme. 
    As discussed above, the maximal $v^*$ is that which induces $\Nbd{v^*}$ that has all paths between any possible $s$ and $t$ flowing through $v^*$. 
    Any $K$-type nodes, by definition, do not pass through $v^*$, so they can be ignored at the maximum value ($n_K(v^*) = 0)$.
    This means the optimal normalization constant is $\pnorm(v^*) = \pnormnok(v^*) = \degout{v^*}\degin{v^*}$ independent of the inclusion of $K$-type terms, and we must only consider the maximal value of $\Nbdnok{v^*}$. 
    The maximum betweenness centrality value is achieved by the neighborhood graph which has all in-nodes to $v^*$ connected only to $v^*$. 
    In such a scenario,
    \begin{align*}
        \BCnorm(v^*) = \BCnormnok(v^*) &= \frac{1}{\pnormnok(v^*)}\sum\limits_{s \in V \setminus \{v^*\}} \sum\limits_{t \in V \setminus \{v^*\}} \shortestp(s,t \mid v^*) \\
        &= \frac{1}{\pnormnok(v^*)} \sum\limits_{s \in V \setminus \{v^*\}} \degout{v^*} \\
        &= \frac{1}{\pnormnok(v^*)} \degout{v^*} \degin{v^*} \\
        &= 1
    \end{align*} 
\end{proof}

As shown in Proposition~\ref{prop:betweenness_normalization}, the $[0,1]$-normalizing constant for graphs of type $\Nbd{v}$ ($\Nbdnok{v}$) is the denominator of the CD Index (no-$k$ CD Index) scaled by $\degout{v}$, and $\CD(v^*) = \BCnorm(v^*) = 1$. 
This observation motivates us to consider to what extent the CD Index and betweenness centrality are related. 
Careful observation of Definitions~\ref{def:betweenness_centrality} and~\ref{def:cd_index} implies that the CD Index is not simply a scaled version of betweenness centrality, as $\CD$ is a difference of label counts over the citing nodes, whereas $\BC$ is a ratio of path counts. 
We can rewrite $\BCnorm$ to further clarify this observation. 
Letting $\pnorm(v) = \degout{v}(n_I(v) + n_J(v) + n_K(v))$ and $T(v) = V(\Dout{v}) \setminus \{v\}$,
\begin{align}
    \BCnorm(v) &= \frac{1}{\pnorm(v)} \left(\sum\limits_{s_i \in I(v)} \sum\limits_{t \in T(v)} \shortestp(s_i,t \mid v) + \sum\limits_{s_j \in J(v)}\sum\limits_{t \in T(v)} \shortestp(s_j,t \mid v) + \sum\limits_{s_k \in K(v)} \sum\limits_{t \in T(v)} \shortestp(s_k,t \mid v)\right)\nonumber \\
    &= \frac{1}{\pnorm(v)} \left(\degout{v}n_I(v) + \sum\limits_{s_j \in J(v)}(\degout{v} - \degout{s_j} + 1) \right)\nonumber \\
    &= \frac{1}{\pnorm(v)} \left(\degout{v}n_I(v) + \degout{v}n_J(v)+ (n_J(v) - \sum\limits_{s_j \in J(v)}\degout{s_j} \right)\nonumber \\
    &= \frac{1}{\pnorm(v)} \left(\degout{v}(n_I(v) + n_J(v)) + n_J(v) - \sum\limits_{s_j \in J(v)}\degout{s_j}) \right)\nonumber \\
    &= \frac{n_I(v) + n_J(v)}{n_I(v) + n_J(v) + n_K(v)} + \frac{n_J(v) - \sum\limits_{s_j \in J(v)} \degout{s_j} }{\degout{v}(n_I(v) + n_J(v) + n_K(v))}\nonumber \\ 
    &= \frac{\degin{v}}{\degin{v} + n_K(v)} - \frac{\sum\limits_{s_j \in J(v)} (\degout{s_j} - 1)}{\degout{v}(\degin{v} + n_K(v))}\label{eq:betweenness_in_cd_notation}.
\end{align}
and by extension when $n_K(v) = 0$:
\begin{equation}\label{eq:betweennessnok_in_cd_notation}
    \BCnormnok(v) = 1 - \frac{\sum\limits_{s_j \in J(v)} (\degout{s_j} - 1) }{\degout{v}\degin{v}}
\end{equation}
% \begin{wrapfigure}{r}{0.5\textwidth}
%     \begin{center}
%     \includegraphics[width=\textwidth]{figures/path_versus_counts.png}
%     \end{center}
%     \caption{Betweenness centrality counts paths, whereas the CD Index counts node types. $\BC(v)$ counts three shortest paths emanating from this $I$-type node, but $\CD(v)$ only counts this $I$-type node once.}
%     \vspace{-20pt}
%     \label{fig:path_versus_counts}
% \end{wrapfigure}
Thus, the betweenness centrality $\BCnormnok(v)$ in $\Nbdnok{v}$ simplifies over this neighborhood graph to a measurement of the relative number of paths in the graph emanating from $J$-type nodes and routing around $v$ to its cited papers. When extended to graph types $\Nbd{v}$ with $K$-type nodes, the resulting value $\BCnorm(v)$ (Equation~\ref{eq:betweenness_in_cd_notation}) becomes the proportion of possible shortest paths that can pass through $v$ minus the number of paths which route around $v$, normalized by the number of possible paths. 

Although they are similar, $\BCnorm$ is not equivalent to $\CD$. 
In fact, $\BCnorm(v)$ and $\CD(v)$ cannot be forced into equality through mere scaling on $\Nbd{v}$ alone.
This is because $\BCnorm(v)$ is a ratio of the number of paths which emanate from $J$-type nodes and therefore has range $[0,1]$, whereas $\CD(v)$ ignores these paths, accounting instead for only the number of $J$-type nodes, resulting in a range $[-1,1]$.

\begin{figure}[htbp]
    \centering
    \includegraphics[width=0.49\textwidth]{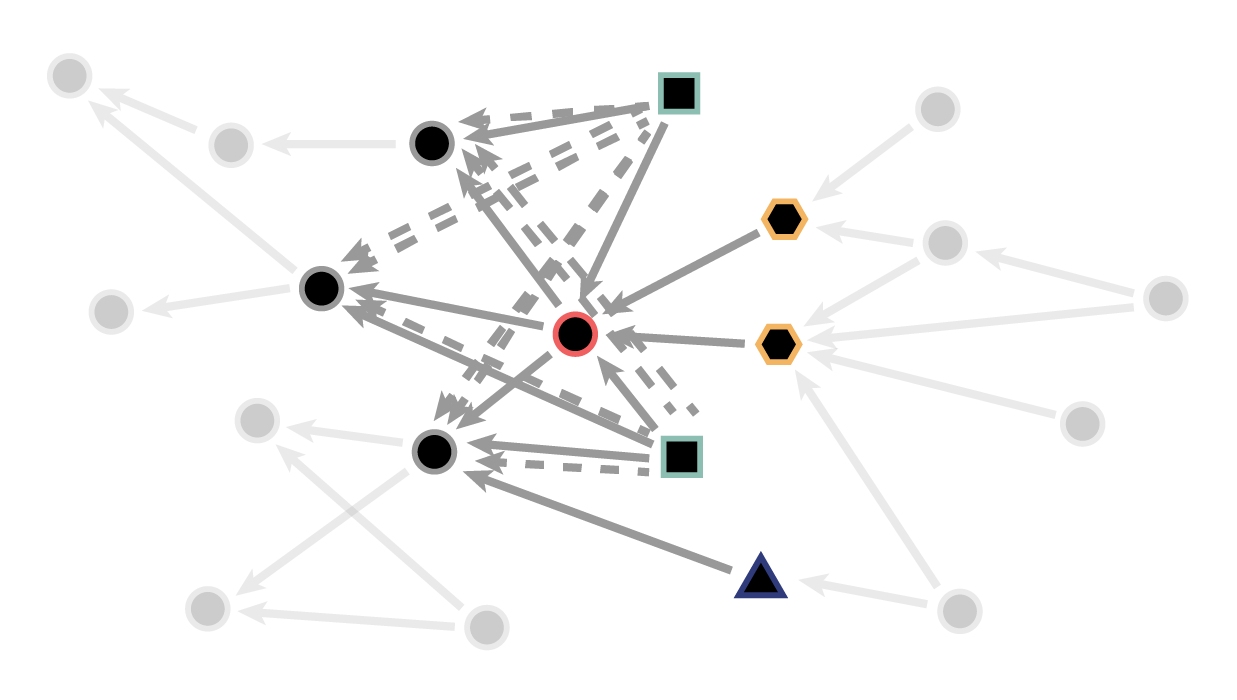}
    \includegraphics[width=0.49\textwidth]{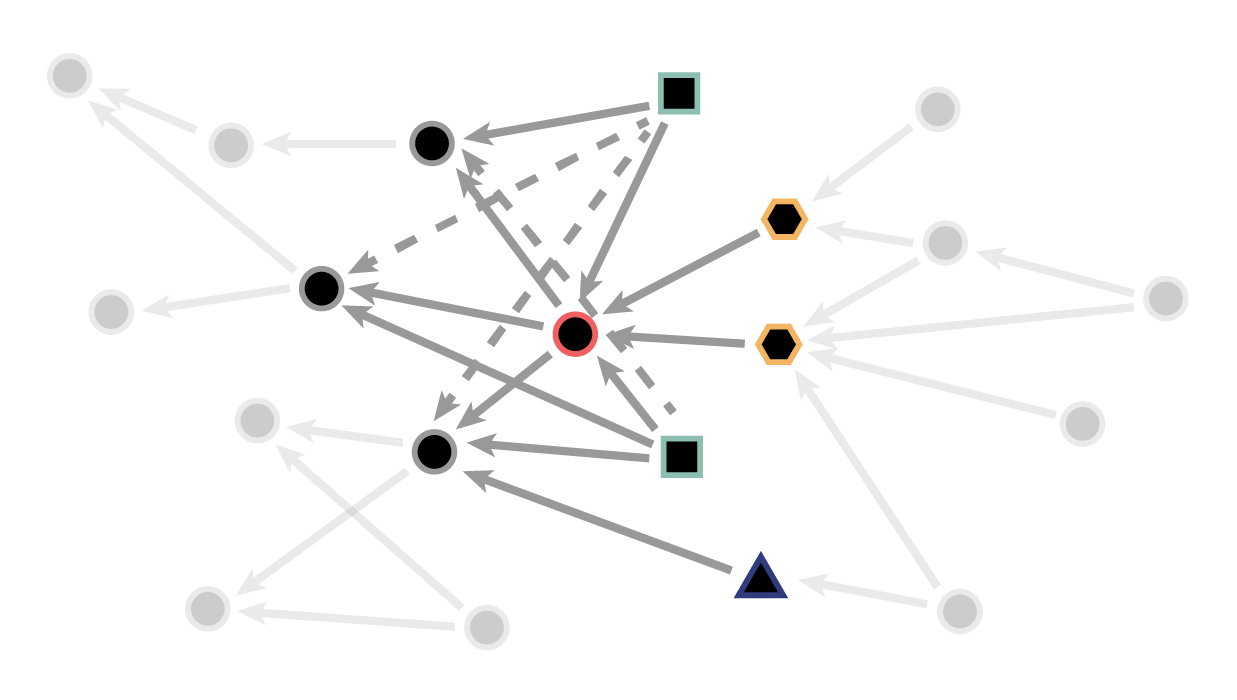}
    \caption{Left: The multi-graph structure which would force $\BCnorm(v) = \CD(v)$. Right: The implicit graph resulting from computing a shifted version of betweenness centrality on $\Nbd{v}$ which is equivalent to $\BCnorm(v) = \CD(v) + q(v)$.}
    \label{fig:multigraph_betweenness}
    \label{fig:shifted_betweenness}
\end{figure}

However, we can force $\BCnorm$ into equivalence with $\CD$ if we are allowed an additional constant to align the ranges. 
To see this, note that setting 
\begin{equation}\label{eq:jout_to_cd}
    \sum\limits_{s_j \in J(v)} \degout{s_j} = (2\degout{v}+1)n_J(v)
\end{equation}
brings Equation~\ref{eq:betweenness_in_cd_notation} into equivalence with $\CD(v)$:
\begin{align*}
    \BCnorm(v) &= \frac{\degin{v}}{\degin{v} + n_K(v)}  - \frac{\sum\limits_{s_j \in J(v)} (\degout{s_j} - 1) }{\degout{v}(\degin{v} + n_K(v)}  \\
    &= \frac{\degin{v}}{\degin{v} + n_K(v)}  - \frac{(2\degout{v} + 1)n_J(v) - n_J(v) }{\degout{v}(\degin{v} + n_K(v)} \\
    &= \frac{\degin{v}}{\degin{v} + n_K(v)}  - \frac{2n_J(v)}{\degin{v} + n_K(v)} \\
    &= \frac{n_I(v) + n_J(v)}{\degin{v} + n_K(v)}  - \frac{2n_J(v)}{\degin{v} + n_K(v)} \\
    &= \frac{n_I(v) - n_J(v)}{\degin{v} + n_K(v)} \\
    &= \CD(v).
\end{align*}

Equation~\ref{eq:jout_to_cd} shows that the CD Index is equivalent to computing the betweenness centrality over a neighborhood graph withing which the average out-degree of $J$-type nodes is $2\degout{v}+1$. 

However, no proper digraph with the structural constraints of $\Nbd{v}$ exists that can satisfy Equation~\ref{eq:jout_to_cd} due to the fact that each node that cites $v$ can out-degree at most $\degout{v}+1$. 
This structural impossibility comes from the fact that $\BCnorm$ is a ratio of paths, whereas $\CD$ is a shifted count of vertices. 
To align these two measures, we can either relax our structural constraints on $\Nbd{v}$ and view it as a multidigraph which has two edges between each $J$-type node and each node cited by $v$, or we can re-shift $\BCnorm(v)$ to align with $\CD(v)$. 
To find this additive constant $q(v)$ for the latter case, set
\begin{align*}
    \BCnorm(v) &= \CD(v) + q(v) \\
\end{align*} and solve for $q(v)$:
\begin{align*}
    q(v) &= \BCnorm(v) - \CD(v) \\
    q(v) &=  \frac{\degin{v}}{\degin{v} + n_K(v)} - \frac{\sum\limits_{s_j \in J(v)} (\degout{s_j} - 1)}{\degout{v}(\degin{v} + n_K(v))} - \frac{n_I(v) - n_J(v)}{\degin{v} + n_K(v)} \\
    q(v) &= \frac{2n_J(v)}{\degin{v} + n_K(v)} - \frac{\sum\limits_{s_j \in J(v)} (\degout{s_j} - 1)}{\degout{v}(\degin{v} + n_K(v))}  
\end{align*}
Thus, we see that the addition $q(v)$ bringing the CD Index and betweenness into alignment is one which re-aligns their ranges by adding $2n_J(v)/(\degin{v} + n_K(v))$ to $\CD(v)$ then deflates $\CD(v)$ by the number of paths emanating from $J$-type nodes which still pass through $v$.  
If we take $\sum_{s_j \in J(v)} (\degout{s_j}-1) = \degout{v}n_J(v)$ to be the maximum possible out-degree supported by our assumptions on the neighborhood graph structure $\Nbd{v}$, we find
\begin{align*}
    q(v) &= \frac{2n_J(v)}{\degin{v} + n_K(v)} - \frac{\degout{v}n_J(v)}{\degout{v}(\degin{v} + n_K(v))}  \\
    &= \frac{n_J(v)}{\degin{v} + n_K(v)}.
\end{align*}

In other words, \textbf{computing the CD Index $\CD(v)$ is equivalent to computing the betweenness centrality} $\BCnorm(v)$ \textbf{on the graph} $\Nbd{v}$ \textbf{where all $J$-type papers in} $\Nbd{v}$ \textbf{cite \emph{all} of $v$'s cited nodes} minus a constant which aligns their ranges by accounting for the proportion of $J$-type nodes in the graph.
Equivalently, on a neighborhood graph wherein all $J$-type nodes connect to each of $v$'s out-neighbors, 
$\BCnorm(v) = n_I(v)/(\degin{v} + n_K(v))$ which measures the proportion of $I$-type nodes to total nodes which is equivalent to the $DI^*$ reformulation of the CD Index presented in \citet{leydesdorff2021proposal} and \citet{chen2021destabilization}. 

\section{Generalizing Disruption Measures via Centrality}~\label{sec:generalizing_disruption}

The relationship the CD Index and betweenness centrality leads to a number of exciting implications regarding the measurement of disruption.
As noted in Section~\ref{sec:cd_index}, the CD Index fails to satisfy Property~\ref{enum:indirect_influence} due to the fact that it is only defined over $\Nbd{v}$ which is composed of only the immediate predecessor and successor works of $v$.
This specificity in citation context can lead to issues in accurately measuring disruption.
For example, if future works do not directly cite paper $v$ and instead attribute better-known or more refined follow-up works of $v$, $\CD(v)$ will be blind to these indirect attributions.
Similarly, if multiple papers compose a disruptive stream of work wherein each subsequent paper builds upon and eclipses the last in relevance, the most recent work is likely to garner the lion's share of disruption as measured by the $\CD(v)$, even though each earlier work composes a piece of the disruptive whole.

Because they may be defined with respect to arbitrarily-sized neighborhood graphs, centrality measures are sensitive to this form of indirect influence and therefore offer a theoretical basis for reasoning about these indirect influences and expanding the notion of disruption to account for such behavior. 
As detailed in Section~\ref{sec:centrality}, betweenness and Pagerank are defined with respect to arbitrary-sized neighborhood subgraphs and therefore satisfy Property~\ref{enum:indirect_influence} by default. 
The relationship between betweenness and the CD Index derived in the previous section motivates one to consider disruption indices constructed from subgraphs of the focal paper of various size and structure, up to the entire ambient citation network. 
Although work has begun to emerge towards this end~\citep{yang2023disruptive}, the heretofore lack of network-theoretic grounding leads to ambiguity in the implementation and properties of the resulting multi-hop measure. 
The centrality framework for disruption presented in this work provides a much more direct route for such measurement generalizations. 

\begin{figure}[htbp]
    \centering
    \includegraphics[width=0.49\textwidth]{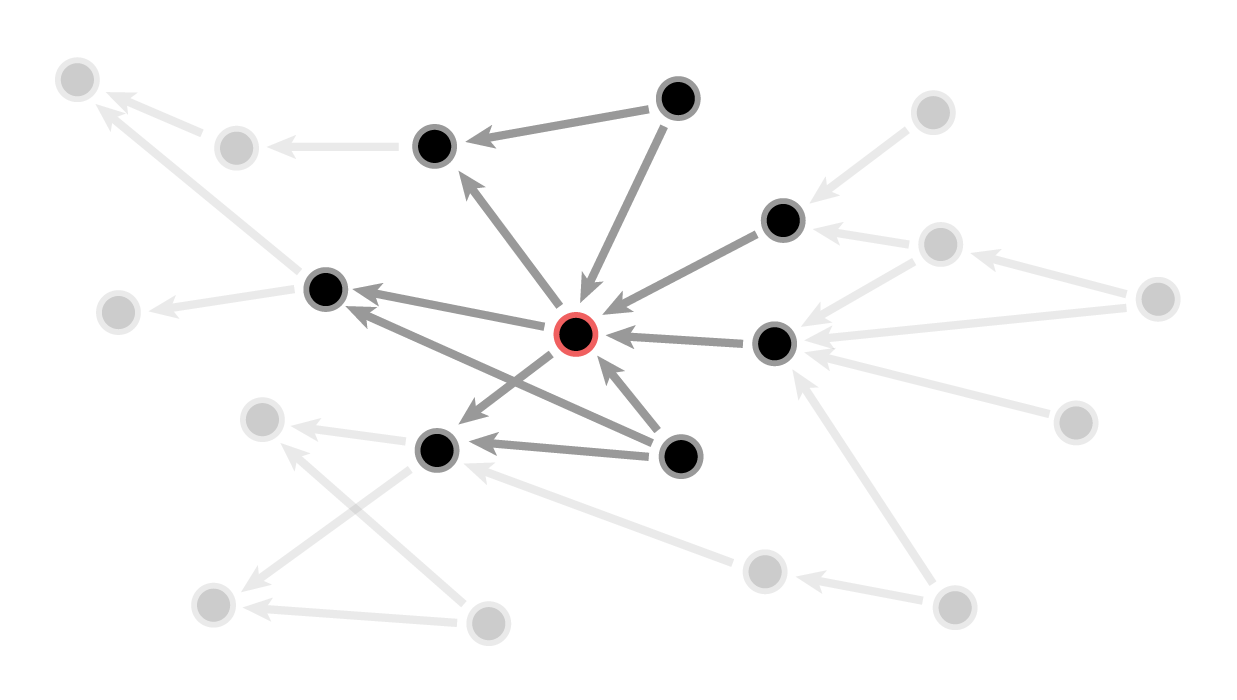}
    \includegraphics[width=0.49\textwidth]{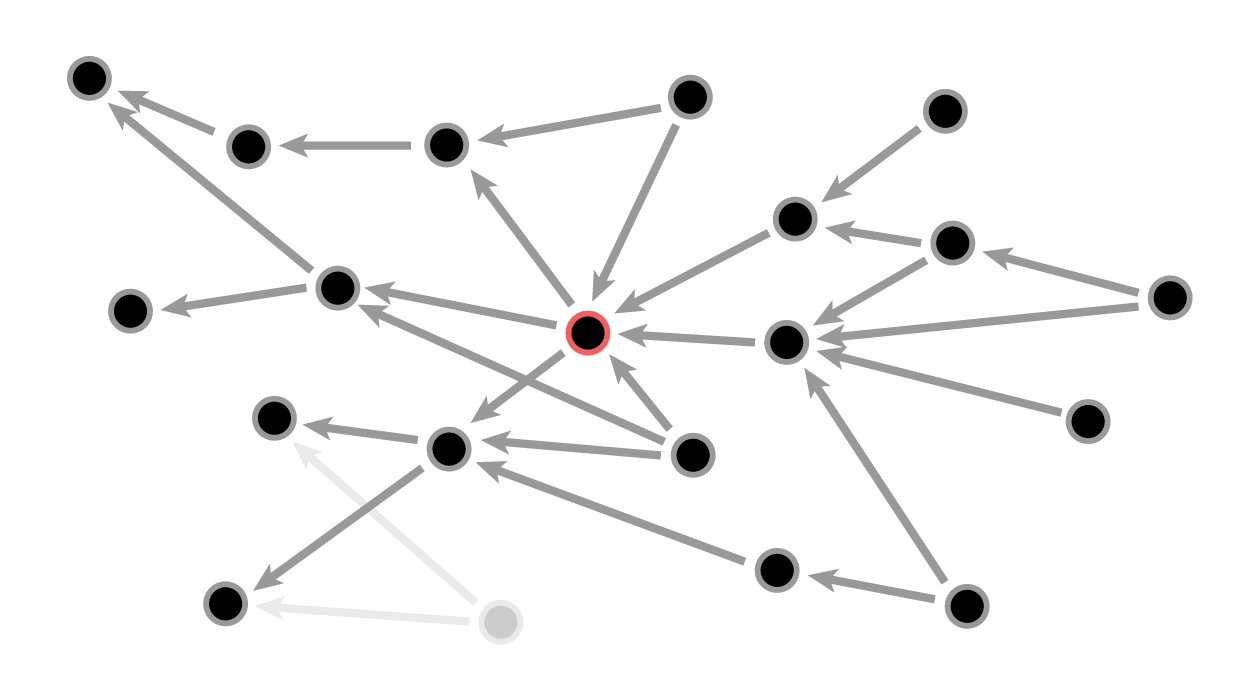}
    \caption{Left: 1-hop ego subgraph $\overline{\mathcal{N}}_1(v)$ around $v$ (red). Right: two-hop ego subgraph $\overline{\mathcal{N}}_2(v)$ around $v$ (red).}
    \label{fig:k_hop_neighborhoods}
\end{figure}

This relationship between centrality measures and disruption also highlights the importance of measurement semantics regarding node importance as measured by centrality. 
As observed in Section~\ref{sec:cd_as_betweenness}, the CD Index counts $J$-type nodes the same, regardless of whether they cite all of the focal paper's prior work or a single prior work. 
By contrast, betweenness is sensitive to the number of paths which route from successors to predecessors of the focal paper.
Pagerank's random walk semantics are also sensitive to paths between citing and cited papers, though the randomness softens the sensitivity to the path length of these walks. 
Such measurement semantics must be taken into account when choosing a disruption measure, whether it be the CD Index, betweenness, Pagerank, or some other centrality measure extended to the disruption context. 

Note that these centrality definitions of disruption also provide computational benefits when measuring disruption over large citation networks. 
Betweenness centrality is typically computed by running an all-pairs shortest path algorithm which benefits from the compositionality of geodesic distances when run on the entire graph at once. 
The dynamic nature of this computation, when computed over the entire network at once, provides significant computational savings as one can avoid computing shortest paths along the same path multiple times for each choice of focal paper. 
This is in contrast to the CD Index and hop-based disruption measures which must compute disruption measures over each subgraph independently, without borrowing information from past computations within the citation network. 
The Pagerank algorithm enjoys similar computational benefits to betweenness when computed on the entire citation network at once, and the Eigenvalue problem in Equation~\ref{eq:pagerank} can be efficiently computed using a power iteration method which allows for arbitrary precision, although convergence issues may warrant consideration~\citep{langville2004deeper}.  
In addition, efficient implementations of these centrality-based algorithms exist across many software packages, and approximation algorithms for betweenness also exist~\citep{brandes2008variants}.

Finally, we note that this idea of using centrality to measure paper importance within citation networks is not new. 
Many past works have investigated the use of centrality measures--especially Pagerank--in highlighting important papers within scientific corpora~\citep{ma2008bringing,maslov2008promise,frahm2014google}, in addition to measuring the relevance of scientisits within their collaboration networks~\citep{senanayake2015pagerank}. 
By explicitly tying disruption to centrality in this work, we can both begin interpreting these past results within the context of disruption, and further extend the study of disruption to other areas of network social science through the shared language of centrality. 

\section{Measuring Disruption in Physics}

This section provides empirical validation for many of the claims made in the previous sections regarding the relationship between disruption and centrality measures. 
Using a 2021 snapshot of the American Physical Society (APS) bibliographic database, containing over 630,000 papers published in APS journals between the years 1893 and 2019, we derive citation networks based on the corresponding citation data.
For each year $t \in [1900,2010]$ of the data, we create a citation network $\Cnet_{t+h}$ representing the citation network of all papers published up to and including year $t+h$ where $h$ is an integer-valued lookahead time horizon. 
For the experiments below, we take $h \in \{5,10,300\}$ where $h=300$ is an ``all-time'' horizon which results in a citation network constructed from all papers in the database.

\begin{figure}[htbp]
    \centering
    \includegraphics[width=0.49\textwidth]{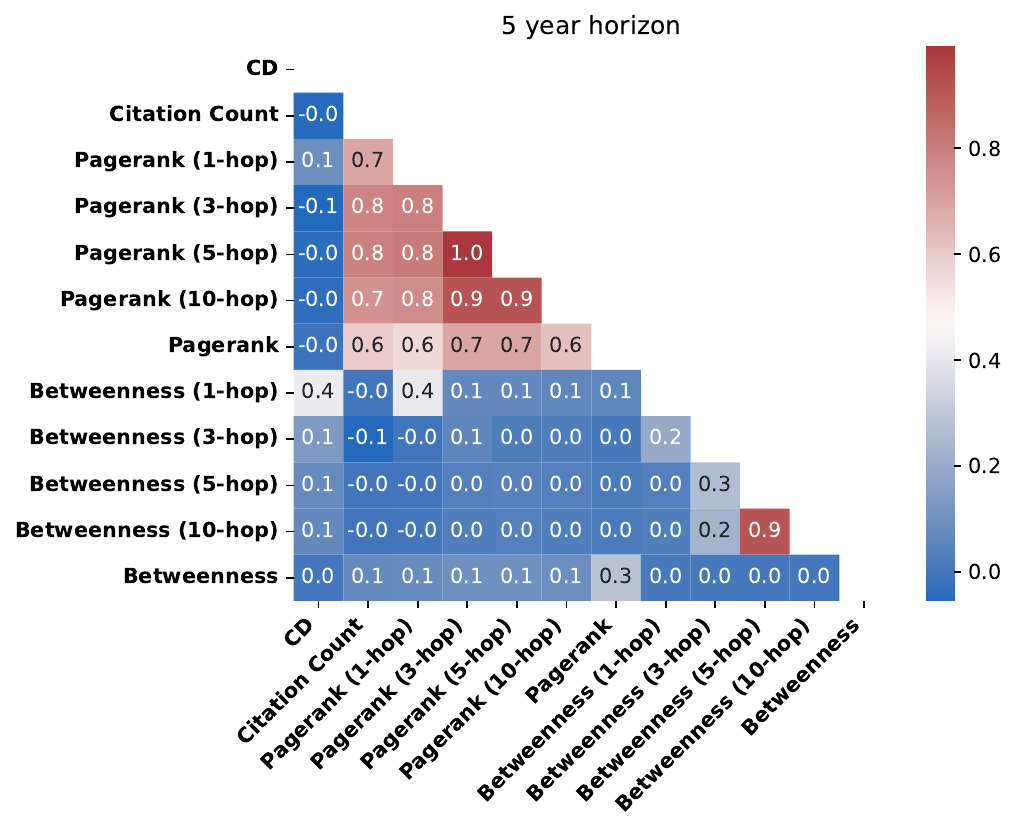}
    \includegraphics[width=0.49\textwidth]{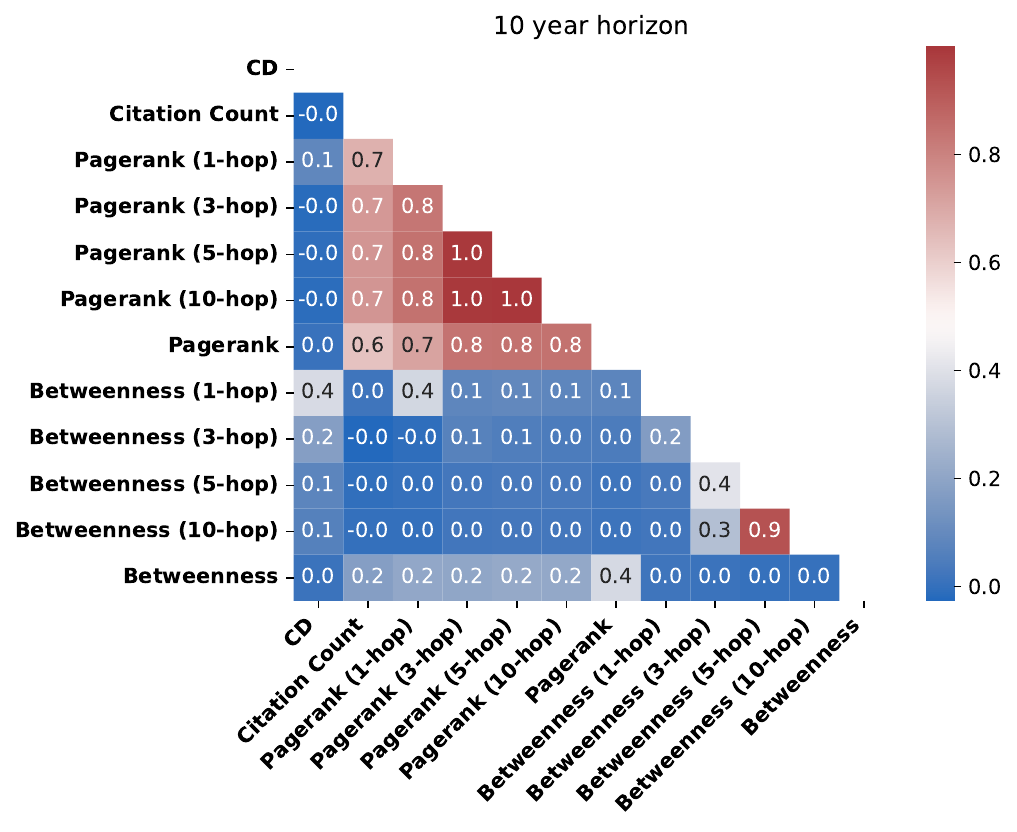}
    \caption{Correlation between centrality measures evaluated at time horizons $h=5$ (left) and $h=10$ (right). Centrality measures without the $k$-hop suffix are computed with respect to the entire network $\Cnet_{t+h}$.}
    \label{fig:correlation}
\end{figure}

For each year $t$ and each time horizon $h$, we compute each paper the citation count ($\IDC$), CD Index ($\CD$), betweenness centrality ($\BC_k$), and Pagerank ($\PR_k$) of each paper published in year $t$ with respect to $\Cnet_{t+h}$. 
For betweenness and Pagerank, we also compute these centrality measures for a range of $k$-hop ego subgraphs $\overline{\mathcal{N}}_k(v)$ for $k \in \{1,3,5,10\}$ in addition to computing centrality over the entire $\Cnet_{t+h}$ network irrespective of any node neighborhood (``all'').
Note again that $\overline{\mathcal{N}}_1(v)$ is not equivalent to $\Nbd{v}$ in general, so we should not expect $\BC_1$ to correlate perfectly with $\CD$. 
We normalize $\BC_k(v)$ by setting $p = (|V(\overline{\mathcal{N}}_k(v))|-1)(|V(\overline{\mathcal{N}}_k(v))|-2)$. 
We normalize $\PR_k(v)$ by dividing each Pagerank value by $\alpha/|V(\overline{\mathcal{N}}_k(v))|$, the lower bound of scores over the neighborhood graph $\overline{\mathcal{N}}_K(v)$. 
We set $\alpha = 0.1$ for Pagerank and let $\vec{\gamma} = \vec{1}|V(\overline{\mathcal{N}}_k(v))|^{-1}$.

\subsection{Correlation among Disruption Measures}

Figure~\ref{fig:correlation} displays the correlation between each disruption measure across 5- and 10-year time horizons and across various $k$-hop subgraphs. 
As expected, the correlation between $\BC_1$ and $\CD$ is substantial, though their disruption measurements are still distinct due to the difference between $\mathcal{N}^{\mathrm{CD}}$ and $\overline{\mathcal{N}}_1$.
Plotting $|\BC_1 - \CD|$ versus $\CD$ as in Figure~\ref{fig:CD-betweenness} provides further empirical evidence for the source of the measurement difference between the CD Index and 1-hop betweenness. 
Overwhelmingly, the papers with $\CD$ values near zero are those with divergent $\BC_1$ values. 

As noted in Section~\ref{sec:cd_as_betweenness}, this is expected due to the fact that a single citation to the focal paper's prior work reduces the value of the CD Index the same as if this paper cited all of the focal paper's prior work, while the same does not hold for betweenness. 
This discontinuity of the CD Index near $\CD=0$ has been noted in prior work~\citep{wu2019confusing}, and the analyses presented here further enforce the notion that disruption based on path measures may provide a more intuitive measure of disruption than those based on node counts. 

\begin{wrapfigure}{r}{0.5\textwidth}
    \begin{center}
    \includegraphics[width=\textwidth]{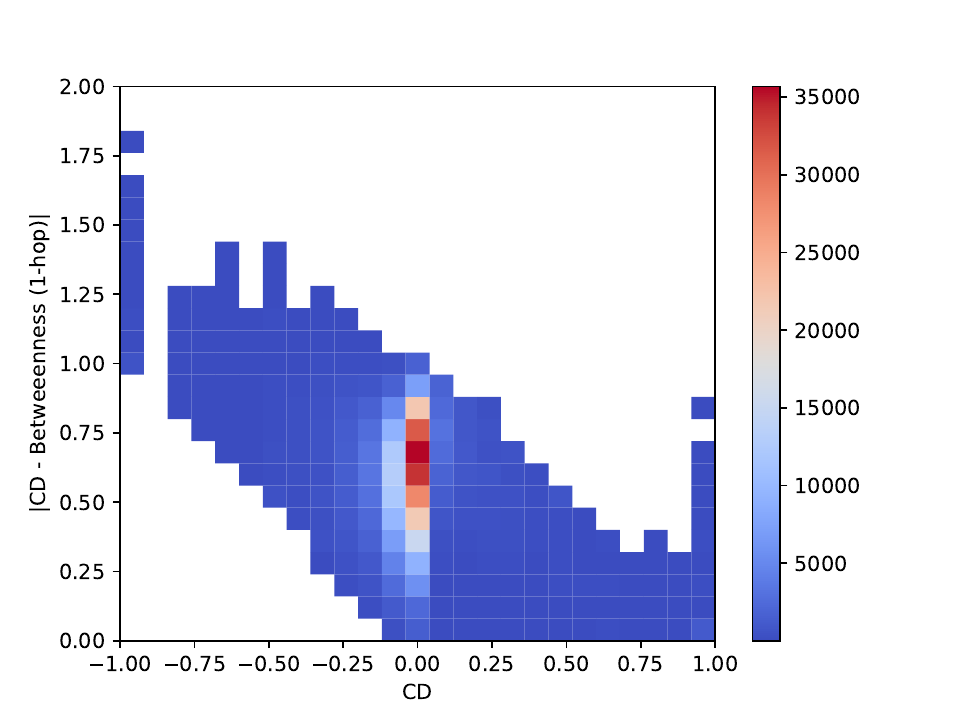}
    \end{center}
    \caption{Histogram of the difference between the CD Index and 1-hop betweenness across all papers in APS corpus versus the CD Index. The histogram is thresholded so each bin has at least 10 papers.}
    \vspace{-20pt}
    \label{fig:CD-betweenness}
\end{wrapfigure}

Besides $\BC_1$, the correlation between $\CD$ and all other disruption measures is limited. This suggests that each of these disruption measures captures varying facets of disruption with respect to each paper's neighborhood network.
The betweenness disruption measures $\BC_k$ diverge from $\CD$ as their neighborhood hop size $k$ increases. 
This divergence is expected, as the options for shortest routes between the past and future papers relative to the focal paper increases combinatorially with each additional hop. 
Intriguingly, the family of betweeenness measures also exhibits relatively low inter-family correlation across values of $k$ except for the pair $\BC_5$ and $\BC_{10}$.
This observation reinforces the concept that the geodesic distance underlying betweenness may exhibit discontinuous behavior when confronted with minor modifications to the structural composition of the neighborhood graph.

By contrast, $\PR_k$ yields disruption measures which display substantial inter-family correlation across all values of $k$, reflecting the intuition that disruption derived from random walks may be more robust to modifications to the structural composition of the neighborhood graph surrounding the focal paper.
Figure~\ref{fig:correlation} also indicates that Pagerank is, in general, highly correlated with citation count. 
This relationship is expected, as a higher proportion of nodes directly citing the focal paper implies a higher likelihood of visiting the focal paper along an arbitrary random walk on $\overline{\mathcal{N}}_k(v)$. 
This relationship is especially relevant when $k=1$, as $\PR_k(v)$ will be primarily driven by the number of papers citing $v$ unless each of these papers cite a significant amount of $v$'s prior work.

Taken together, the results of these correlation analyses reinforce the observation from Section~\ref{sec:generalizing_disruption} that the definition and measurement of node ``importance'' over the citation network is integral to the resulting semantics of these disruption measures.
The node-type importance of $\CD$ diverges from the geodesic betweenness measure of importance underlying $\BC_k$ as $k$ increases. 
Similarly, the random walk visitation importance underlying $\PR_k$ results in a disruption measure which is correlated to both Citation Count and $\BC_1$ when $k=1$, but then diverges from the latter measure with increasing $k$ until all hops are taken into account. 
Although each of these measures appears to highlight a distinct facet of disruption, we will observe in the next section that they do share similarity in their aggregate trends across time.

\subsection{Disruption Trends}

Since its introduction, the CD Index has seen frequent employment as a measure of disruptive outcomes with respect to particular structural variables relevant to the science of science like team size~\citep{wu2019large} or scientific discourse~\citep{lin2022new}. 
A recent study by ~\citet{park2023papers} featured the CD Index in an evaluation of the slowing pace of scientific and technological innovation. 
In this work, the authors tracked the yearly average value of the CD Index across time and found a generally decreasing trend in disruptiveness measured within citation networks across the sciences. 
Figure~\ref{fig:disruption_over_time} plots the yearly average disruption across papers and their correlations, smoothed by averaging across a centered 5-year window, evaluated at a 5-year time horizon ($h=5$).  
Similar to the observations in ~\citet{park2023papers}, we observe a generally downward trend in disruption through time across all alternative specifications of disruption apart from Citation Count and $\PR_1$.
This latter outlier similarity between 1-hop Pagerank and Citation Count is expected due to their close theoretical relationship stemming from a strong influence of direct citations on these measures.
As the neighborhood expands, however, we see the time series of Pagerank quickly begins to mirror that of $\CD$ and $\BC$. 
As the correlation plot suggests, these alternative disruption specifications are highly similar in terms of their aggregate measurements of disruption through time. 
The high variance of the average yearly value of betweenness across time is an interesting artifact deserving of further study. 
We hypothesize that these spikes may correspond to structural shifts in the topology of the citation network which, due to the geodesic distance underlying betweenness, results in large jumps in average betweenness across the network.
Such structural changes to the citation network topology may be scientifically meaningful and therefore deserve further study.
These results imply that even though the alternative measurements of disruption introduced in this paper may see low correlation to the CD Index at the paper level, these alternative specifications are still influenced by similar global trends in the structure of the citation network which play out through time. 

\begin{figure}[htbp]
    \centering
    \includegraphics[width=0.64\textwidth]{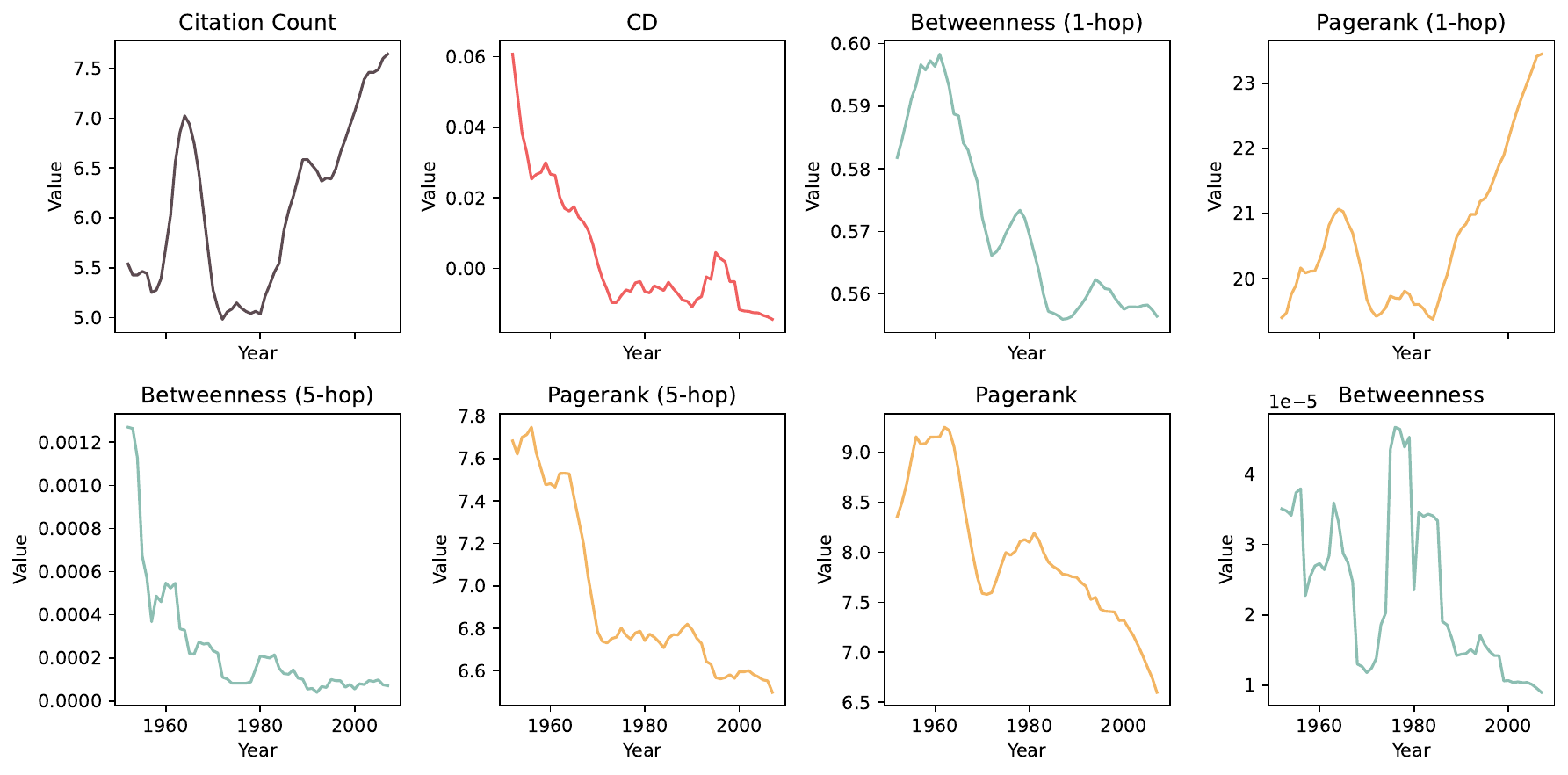}
    \includegraphics[width=0.35\textwidth]{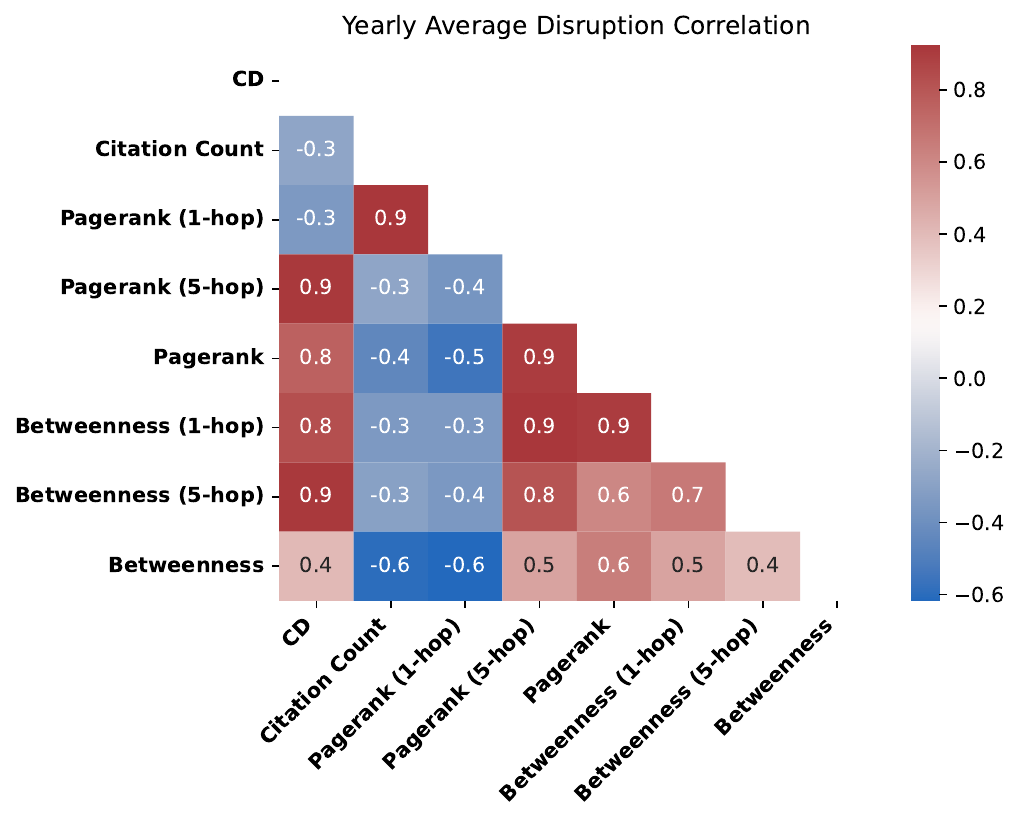}
    \caption{Left: Disruption measures for APS plotted between 1950 and 2010 with $h=5$. Each time series represents the average disruption value computed within each year smoothed by averaging across the two prior and succeeding years. Right: The correlation among each of these time series.}
    \label{fig:disruption_over_time}
\end{figure}

\subsection{Measuring Prize-Winning Papers}

% Please add the following required packages to your document preamble:
% \usepackage{booktabs}
% \usepackage{multirow}
\begin{table}[htbp]
\begin{tabular}{@{}lllll@{}}
\toprule
\textbf{Disruption Type}                 & $\bm{k}$                   & $\bm{h}$  & \textbf{AMR}    & \textbf{stdev}  \\ \midrule
\multirow{2}{*}{\textbf{CD}}               & \multirow{2}{*}{1}  & 5  & 45.7\% & 40.0\% \\
                                  &                     & 10 & 40.9\% & 39.1\%
                                  \\ \cmidrule(l){1-5}
\multirow{2}{*}{\textbf{Citation Count}} & \multirow{2}{*}{1}  & 5  & 9.9\%  & 18.9\% \\
                                  &                     & 10 & 9.3\%  & 18.2\%
                                  \\ \cmidrule(l){1-5}
\multirow{10}{*}{\textbf{Betweenness}}     & \multirow{2}{*}{1}  & 5  & 33.0\% & 22.1\% \\
                                  &                     & 10 & 30.3\% & 21.6\%
                                  \\ \cmidrule(l){2-5}
                                  & \multirow{2}{*}{3}  & 5  & 19.4\% & 13.7\% \\
                                  &                     & 10 & 13.1\% & 10.6\%
                                  \\ \cmidrule(l){2-5}
                                  & \multirow{2}{*}{5}  & 5  & 4.7\%  & 8.3\%  \\
                                  &                     & 10 & 5.4\%  & 8.0\%
                                  \\ \cmidrule(l){2-5}
                                  & \multirow{2}{*}{10} & 5  & 4.7\%  & 9.0\%  \\
                                  &                     & 10 & 2.8\%  & 5.0\%
                                  \\ \cmidrule(l){2-5}
                                  & \multirow{2}{*}{all}  & 5  & 3.6\%  & 7.0\%  \\
                                  &                     & 10 & 3.3\%  & 6.9\%
                                  \\ \cmidrule(l){1-5}
\multirow{10}{*}{\textbf{Pagerank}}        & \multirow{2}{*}{1}  & 5  & 11.8\% & 20.0\% \\
                                  &                     & 10 & 8.9\%  & 18.7\%
                                  \\ \cmidrule(l){2-5}
                                  & \multirow{2}{*}{3}  & 5  & 5.6\%  & 14.9\% \\
                                  &                     & 10 & 4.5\%  & 12.0\%
                                  \\ \cmidrule(l){2-5}
                                  & \multirow{2}{*}{5}  & 5  & 5.3\%  & 14.6\% \\
                                  &                     & 10 & 4.2\%  & 11.9\%
                                  \\ \cmidrule(l){2-5}
                                  & \multirow{2}{*}{10} & 5  & 6.0\%  & 14.9\% \\
                                  &                     & 10 & 4.1\%  & 11.8\%
                                  \\ \cmidrule(l){2-5}
                                  & \multirow{2}{*}{all}  & 5  & 3.2\%  & 9.0\%  \\
                                  &                     & 10 & 3.6\%  & 12.3\%  \\ \cmidrule(l){1-5} 
\end{tabular}
\caption{Average mean rank (AMR) of Nobel prize-winning papers according to various disruption measures evaluated 5 and 10 years after publication.}
\label{tab:amr}
\end{table}

Using data compiled from ~\cite{li2019dataset}, we determined all papers in the APS database which were cited within acceptance lectures for the Nobel prize in physics. 
We then computed the descending percentile rank of each paper in the database according to a each disruption measure. 
Based on this percentile ranking, we computed the mean ranking of all Nobel prize-winning papers within the dataset. 
The mean ranking of Nobel prize-winning papers across each time horizon and across each choice of $k$ (for $\BC_k$ and $\PR_k$) is plotted in Figure~\ref{fig:mean_ranking}.

\begin{figure}[htbp]
    \centering
    \includegraphics[width=0.87\textwidth]{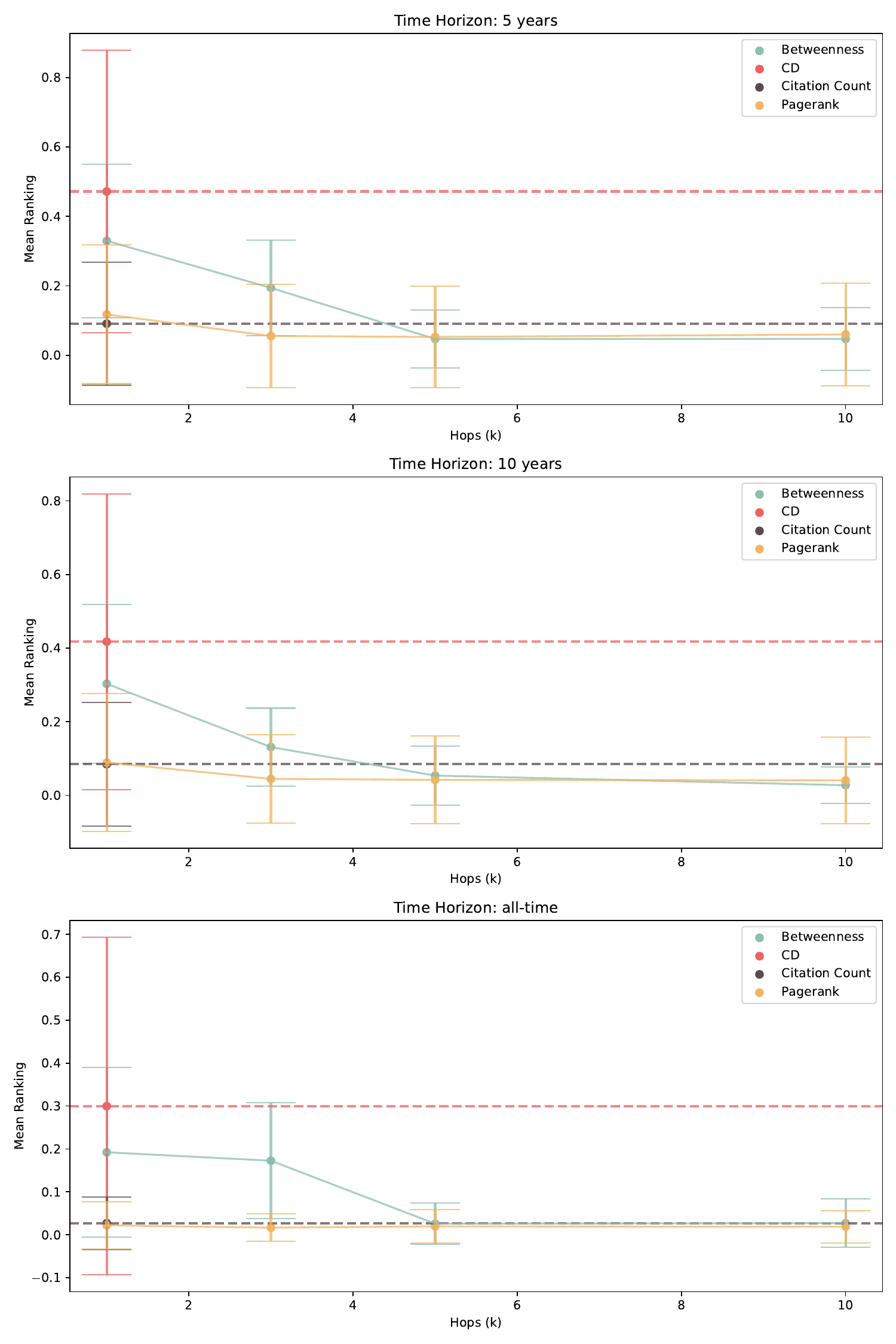}
    \caption{Mean disruption ranking of Nobel prize-winning papers all other papers. Ranking is reported as ``top-p'' percentile where p is the ranking in percentage terms.}
    \label{fig:mean_ranking}
\end{figure}

We observe that the CD Index generally ranks Nobel prize-winning papers as being only slightly more disruptive than the median paper across 5- and 10-year time horizons with high variance in this ranking, and 
unsurprisingly, the $1$-hop betweenness centrality measure of disruption $\BC_1$ performs similarly. 
However, we find that $\BC_k$ assigns higher rank to prize-winning papers as $k$ increases while the variance of this ranking shrinks.
As seen in Table~\ref{tab:amr}, computing betweenness on an arbitrary number of hops provides the highest ranking to Nobel prize-winning papers with the lowest standard deviation, followed by Pagerank and $\BC_{10}$ for 5- and 10-year time horizons.
However as the number of hops decreases, the average ranking assigned by $\BC_k$ approaches 50\%, whereas the $\PR_k$ rankings experience less average rank deterioration but with increased spread. 

\section{Summary}
We have observed that network centrality satisfies many of the properties desired by a measure of scientific and technological disruption within citation networks. 
In addition, we have shown that citation count and variants of the CD Index may be explicitly rewritten as measuring degree and betweenness centrality, respectively. 
By detailing this link between citation disruption and network centrality, and by explicitly consolidating existing disruption measures as specific forms of centrality, we have provided a more general and mathematically rigorous framework for measuring scientific and technological disruption within citation networks.

We have provided empirical evidence which reinforces the value of this network centrality view of disruption by observing the disruption assigned to Nobel prize-winning papers under various definitions of disruption and paper contexts. 
In particular, we find that disruption measurements may benefit from an expanded the context outside of a one-hop neighborhood around the focal paper which provides better accounts for down-stream innovations which may not directly attribute the focal work. 
We hope this relationship between citation disruption and network centrality will foster a more fruitful exchange of ideas between the field of science and innovation studies and the broader network science community.

\bibliographystyle{plainnat}
\bibliography{disruption}

\begin{thebibliography}{62}
\providecommand{\natexlab}[1]{#1}
\providecommand{\url}[1]{\texttt{#1}}
\expandafter\ifx\csname urlstyle\endcsname\relax
  \providecommand{\doi}[1]{doi: #1}\else
  \providecommand{\doi}{doi: \begingroup \urlstyle{rm}\Url}\fi

\bibitem[Andrade et~al.(2020)Andrade, Figueiredo, Silva, and
  Morais]{andrade2020measuring}
Felipe Falc{\~a}o1~Nazareno Andrade, Fl{\'a}vio Figueiredo, Diego Silva, and
  Fabio Morais.
\newblock Measuring disruption in song similarity networks.
\newblock In \emph{Proc. of the 21st Int. Society for Music Information
  Retrieval Conf}, 2020.

\bibitem[Anthonisse(1971)]{anthonisse1971rush}
Jac~M Anthonisse.
\newblock The rush in a directed graph.
\newblock \emph{Stichting Mathematisch Centrum. Mathematische Besliskunde},
  1971.

\bibitem[Arthur(2007)]{arthur2007structure}
W~Brian Arthur.
\newblock The structure of invention.
\newblock \emph{Research Policy}, 36\penalty0 (2):\penalty0 274--287, 2007.

\bibitem[Arthur(2009)]{arthur2009nature}
W~Brian Arthur.
\newblock \emph{The nature of technology: What it is and how it evolves}.
\newblock Simon and Schuster, 2009.

\bibitem[Azoulay et~al.(2019)Azoulay, Fons-Rosen, and
  Graff~Zivin]{azoulay2019does}
Pierre Azoulay, Christian Fons-Rosen, and Joshua~S Graff~Zivin.
\newblock Does science advance one funeral at a time?
\newblock \emph{American Economic Review}, 109\penalty0 (8):\penalty0
  2889--2920, 2019.

\bibitem[Bavelas(1950)]{bavelas1950communication}
Alex Bavelas.
\newblock Communication patterns in task-oriented groups.
\newblock \emph{The journal of the acoustical society of America}, 22\penalty0
  (6):\penalty0 725--730, 1950.

\bibitem[Bloch et~al.(2023)Bloch, Jackson, and Tebaldi]{bloch2023centrality}
Francis Bloch, Matthew~O Jackson, and Pietro Tebaldi.
\newblock Centrality measures in networks.
\newblock \emph{Social Choice and Welfare}, pages 1--41, 2023.

\bibitem[Borgatti et~al.(2009)Borgatti, Mehra, Brass, and
  Labianca]{borgatti2009network}
Stephen~P Borgatti, Ajay Mehra, Daniel~J Brass, and Giuseppe Labianca.
\newblock Network analysis in the social sciences.
\newblock \emph{science}, 323\penalty0 (5916):\penalty0 892--895, 2009.

\bibitem[Bornmann(2020)]{bornmann2020can}
Lutz Bornmann.
\newblock How can citation impact in bibliometrics be normalized? {A} new
  approach combining citing-side normalization and citation percentiles.
\newblock \emph{Quantitative Science Studies}, 1\penalty0 (4):\penalty0
  1553--1569, 2020.

\bibitem[Bornmann and Daniel(2008)]{bornmann2008citation}
Lutz Bornmann and Hans-Dieter Daniel.
\newblock What do citation counts measure? a review of studies on citing
  behavior.
\newblock \emph{Journal of documentation}, 64\penalty0 (1):\penalty0 45--80,
  2008.

\bibitem[Bornmann and Marx(2015)]{bornmann2015methods}
Lutz Bornmann and Werner Marx.
\newblock Methods for the generation of normalized citation impact scores in
  bibliometrics: Which method best reflects the judgements of experts?
\newblock \emph{Journal of Informetrics}, 9\penalty0 (2):\penalty0 408--418,
  2015.

\bibitem[Bornmann et~al.(2020{\natexlab{a}})Bornmann, Devarakonda, Tekles, and
  Chacko]{bornmann2020disruption}
Lutz Bornmann, Sitaram Devarakonda, Alexander Tekles, and George Chacko.
\newblock Are disruption index indicators convergently valid? the comparison of
  several indicator variants with assessments by peers.
\newblock \emph{Quantitative Science Studies}, 1\penalty0 (3):\penalty0
  1242--1259, 2020{\natexlab{a}}.

\bibitem[Bornmann et~al.(2020{\natexlab{b}})Bornmann, Devarakonda, Tekles, and
  Chacko]{bornmann2020disruptive}
Lutz Bornmann, Sitaram Devarakonda, Alexander Tekles, and George Chacko.
\newblock Disruptive papers published in scientometrics: meaningful results by
  using an improved variant of the disruption index originally proposed by wu,
  wang, and evans (2019).
\newblock \emph{Scientometrics}, 123\penalty0 (2):\penalty0 1149--1155,
  2020{\natexlab{b}}.

\bibitem[Brandes(2008)]{brandes2008variants}
Ulrik Brandes.
\newblock On variants of shortest-path betweenness centrality and their generic
  computation.
\newblock \emph{Social networks}, 30\penalty0 (2):\penalty0 136--145, 2008.

\bibitem[Bu et~al.(2021)Bu, Waltman, and Huang]{bu2021multidimensional}
Yi~Bu, Ludo Waltman, and Yong Huang.
\newblock A multidimensional framework for characterizing the citation impact
  of scientific publications.
\newblock \emph{Quantitative science studies}, 2\penalty0 (1):\penalty0
  155--183, 2021.

\bibitem[Chen et~al.(2021)Chen, Shao, and Fan]{chen2021destabilization}
Jiyao Chen, Diana Shao, and Shaokun Fan.
\newblock Destabilization and consolidation: Conceptualizing, measuring, and
  validating the dual characteristics of technology.
\newblock \emph{Research policy}, 50\penalty0 (1):\penalty0 104115, 2021.

\bibitem[Chu and Evans(2021)]{chu2021slowed}
Johan~SG Chu and James~A Evans.
\newblock Slowed canonical progress in large fields of science.
\newblock \emph{Proceedings of the National Academy of Sciences}, 118\penalty0
  (41):\penalty0 e2021636118, 2021.

\bibitem[David(1990)]{david1990dynamo}
Paul~A David.
\newblock The dynamo and the computer: an historical perspective on the modern
  productivity paradox.
\newblock \emph{American Economic Review}, 80\penalty0 (2):\penalty0 355--361,
  1990.

\bibitem[Deng and Zeng(2023)]{deng2023enhancing}
Nan Deng and An~Zeng.
\newblock Enhancing the robustness of the disruption metric against noise.
\newblock \emph{Scientometrics}, 128\penalty0 (4):\penalty0 2419--2428, 2023.

\bibitem[Dosi(1982)]{dosi1982technological}
Giovanni Dosi.
\newblock Technological paradigms and technological trajectories: a suggested
  interpretation of the determinants and directions of technical change.
\newblock \emph{Research policy}, 11\penalty0 (3):\penalty0 147--162, 1982.

\bibitem[Enos(1958)]{enos1958measure}
John~L Enos.
\newblock A measure of the rate of technological progress in the petroleum
  refining industry.
\newblock \emph{The Journal of Industrial Economics}, 6\penalty0 (3):\penalty0
  180--197, 1958.

\bibitem[Enos(1962)]{enos1962invention}
John~L Enos.
\newblock Invention and innovation in the petroleum refining industry.
\newblock In \emph{The rate and direction of inventive activity: Economic and
  social factors}, pages 299--322. Princeton University Press, 1962.

\bibitem[Figueiredo and Andrade(2019)]{figueiredo2019quantifying}
Flavio Figueiredo and Nazareno Andrade.
\newblock Quantifying disruptive influence in the allmusic guide.
\newblock In \emph{ISMIR}, pages 832--838, 2019.

\bibitem[Fleck(2012)]{fleck2012genesis}
Ludwik Fleck.
\newblock \emph{Genesis and development of a scientific fact}.
\newblock University of Chicago Press, 2012.

\bibitem[Fortunato et~al.(2018)Fortunato, Bergstrom, B{\"o}rner, Evans,
  Helbing, Milojevi{\'c}, Petersen, Radicchi, Sinatra, Uzzi,
  et~al.]{fortunato2018science}
Santo Fortunato, Carl~T Bergstrom, Katy B{\"o}rner, James~A Evans, Dirk
  Helbing, Sta{\v{s}}a Milojevi{\'c}, Alexander~M Petersen, Filippo Radicchi,
  Roberta Sinatra, Brian Uzzi, et~al.
\newblock Science of science.
\newblock \emph{Science}, 359\penalty0 (6379):\penalty0 eaao0185, 2018.

\bibitem[Frahm et~al.(2014)Frahm, Eom, and Shepelyansky]{frahm2014google}
Klaus~M Frahm, Young-Ho Eom, and Dima~L Shepelyansky.
\newblock Google matrix of the citation network of physical review.
\newblock \emph{Physical Review E}, 89\penalty0 (5):\penalty0 052814, 2014.

\bibitem[Freeman(1977)]{freeman1977set}
Linton~C Freeman.
\newblock A set of measures of centrality based on betweenness.
\newblock \emph{Sociometry}, pages 35--41, 1977.

\bibitem[Funk and Owen-Smith(2017)]{funk2017dynamic}
Russell~J Funk and Jason Owen-Smith.
\newblock A dynamic network measure of technological change.
\newblock \emph{Management science}, 63\penalty0 (3):\penalty0 791--817, 2017.

\bibitem[Ghasemi et~al.(2014)Ghasemi, Seidkhani, Tamimi, Rahgozar, and
  Masoudi-Nejad]{ghasemi2014centrality}
Mahdieh Ghasemi, Hossein Seidkhani, Faezeh Tamimi, Maseud Rahgozar, and Ali
  Masoudi-Nejad.
\newblock Centrality measures in biological networks.
\newblock \emph{Current Bioinformatics}, 9\penalty0 (4):\penalty0 426--441,
  2014.

\bibitem[Jaffe and Trajtenberg(2002)]{jaffe2002patents}
Adam~B Jaffe and Manuel Trajtenberg.
\newblock \emph{Patents, citations, and innovations: A window on the knowledge
  economy}.
\newblock MIT press, 2002.

\bibitem[Kuhn(1962)]{kuhn1962structure}
Thomas~S Kuhn.
\newblock The structure of scientifi revolutions.
\newblock \emph{The Un of Chicago Press}, 2:\penalty0 90, 1962.

\bibitem[Landherr et~al.(2010)Landherr, Friedl, and
  Heidemann]{landherr2010critical}
Andrea Landherr, Bettina Friedl, and Julia Heidemann.
\newblock A critical review of centrality measures in social networks.
\newblock \emph{Wirtschaftsinformatik}, 52:\penalty0 367--382, 2010.

\bibitem[Langville and Meyer(2004)]{langville2004deeper}
Amy~N Langville and Carl~D Meyer.
\newblock Deeper inside pagerank.
\newblock \emph{Internet Mathematics}, 1\penalty0 (3):\penalty0 335--380, 2004.

\bibitem[Leahey et~al.(2023)Leahey, Lee, and Funk]{leahey2023types}
Erin Leahey, Jina Lee, and Russell~J Funk.
\newblock What types of novelty are most disruptive?
\newblock \emph{American Sociological Review}, 88\penalty0 (3):\penalty0
  562--597, 2023.

\bibitem[Leibel and Bornmann(2023)]{leibel2023we}
Christian Leibel and Lutz Bornmann.
\newblock What do we know about the disruption indicator in scientometrics? an
  overview of the literature.
\newblock \emph{arXiv preprint arXiv:2308.02383}, 2023.

\bibitem[Leydesdorff et~al.(2021)Leydesdorff, Tekles, and
  Bornmann]{leydesdorff2021proposal}
Loet Leydesdorff, Alexander Tekles, and Lutz Bornmann.
\newblock A proposal to revise the disruption index.
\newblock \emph{Profesional de la informaci{\'o}n (EPI)}, 30\penalty0 (1),
  2021.

\bibitem[Li et~al.(2019)Li, Yin, Fortunato, and Wang]{li2019dataset}
Jichao Li, Yian Yin, Santo Fortunato, and Dashun Wang.
\newblock A dataset of publication records for nobel laureates.
\newblock \emph{Scientific Data}, 6\penalty0 (1):\penalty0 1--10, 2019.

\bibitem[Li and Chen(2022)]{li2022measuring}
Jiexun Li and Jiyao Chen.
\newblock Measuring destabilization and consolidation in scientific knowledge
  evolution.
\newblock \emph{Scientometrics}, 127\penalty0 (10):\penalty0 5819--5839, 2022.

\bibitem[Lin et~al.(2022)Lin, Evans, and Wu]{lin2022new}
Yiling Lin, James~A Evans, and Lingfei Wu.
\newblock New directions in science emerge from disconnection and discord.
\newblock \emph{Journal of Informetrics}, 16\penalty0 (1):\penalty0 101234,
  2022.

\bibitem[Liu et~al.(2023)Liu, Jones, Uzzi, and Wang]{liu2023data}
Lu~Liu, Benjamin~F Jones, Brian Uzzi, and Dashun Wang.
\newblock Data, measurement and empirical methods in the science of science.
\newblock \emph{Nature human behaviour}, pages 1--13, 2023.

\bibitem[Ma et~al.(2008)Ma, Guan, and Zhao]{ma2008bringing}
Nan Ma, Jiancheng Guan, and Yi~Zhao.
\newblock Bringing pagerank to the citation analysis.
\newblock \emph{Information Processing \& Management}, 44\penalty0
  (2):\penalty0 800--810, 2008.

\bibitem[Maslov and Redner(2008)]{maslov2008promise}
Sergei Maslov and Sidney Redner.
\newblock Promise and pitfalls of extending google's pagerank algorithm to
  citation networks.
\newblock \emph{Journal of Neuroscience}, 28\penalty0 (44):\penalty0
  11103--11105, 2008.

\bibitem[Mokyr(1992)]{mokyr1992lever}
Joel Mokyr.
\newblock \emph{The lever of riches: Technological creativity and economic
  progress}.
\newblock Oxford University Press, 1992.

\bibitem[Newman(2018)]{newman2018networks}
M.~Newman.
\newblock \emph{Networks}.
\newblock OUP Oxford, 2018.
\newblock ISBN 9780192527493.
\newblock URL \url{https://books.google.com/books?id=YdZjDwAAQBAJ}.

\bibitem[Page et~al.(1998)Page, Brin, Motwani, and Winograd]{page1998pagerank}
Lawrence Page, Sergey Brin, Rajeev Motwani, and Terry Winograd.
\newblock The pagerank citation ranking: Bring order to the web.
\newblock Technical report, Technical report, stanford University, 1998.

\bibitem[Park et~al.(2023)Park, Leahey, and Funk]{park2023papers}
Michael Park, Erin Leahey, and Russell~J Funk.
\newblock Papers and patents are becoming less disruptive over time.
\newblock \emph{Nature}, 613\penalty0 (7942):\penalty0 138--144, 2023.

\bibitem[Popper(2005)]{popper2005logic}
Karl Popper.
\newblock \emph{The logic of scientific discovery}.
\newblock Routledge, 2005.

\bibitem[Price(1963)]{price1963}
Derek J. de~Solla Price.
\newblock \emph{Little science, Big science}.
\newblock Columbia University Press, 1963.

\bibitem[Rosenberg(1982)]{rosenberg1982inside}
Nathan Rosenberg.
\newblock \emph{Inside the black box: Technology and economics}.
\newblock Cambridge University Press, 1982.

\bibitem[Schumpeter(1942)]{schumpeter1942capitalism}
Joseph~A Schumpeter.
\newblock \emph{Capitalism, socialism and democracy}.
\newblock Harper \& Brothers, 1942.

\bibitem[Senanayake et~al.(2015)Senanayake, Piraveenan, and
  Zomaya]{senanayake2015pagerank}
Upul Senanayake, Mahendra Piraveenan, and Albert Zomaya.
\newblock The pagerank-index: Going beyond citation counts in quantifying
  scientific impact of researchers.
\newblock \emph{PloS one}, 10\penalty0 (8):\penalty0 e0134794, 2015.

\bibitem[Tahamtan and Bornmann(2018)]{tahamtan2018core}
Iman Tahamtan and Lutz Bornmann.
\newblock Core elements in the process of citing publications: Conceptual
  overview of the literature.
\newblock \emph{Journal of informetrics}, 12\penalty0 (1):\penalty0 203--216,
  2018.

\bibitem[Tahamtan and Bornmann(2019)]{tahamtan2019citation}
Iman Tahamtan and Lutz Bornmann.
\newblock What do citation counts measure? an updated review of studies on
  citations in scientific documents published between 2006 and 2018.
\newblock \emph{Scientometrics}, 121:\penalty0 1635--1684, 2019.

\bibitem[Tushman and Anderson(1986)]{tushman1986technological}
Michael~L Tushman and Philip Anderson.
\newblock Technological discontinuities and organizational environments.
\newblock \emph{Administrative Science Quarterly}, 31\penalty0 (3):\penalty0
  439--465, 1986.

\bibitem[Waltman(2016)]{waltman2016review}
Ludo Waltman.
\newblock A review of the literature on citation impact indicators.
\newblock \emph{Journal of informetrics}, 10\penalty0 (2):\penalty0 365--391,
  2016.

\bibitem[Wang and Barab{\'a}si(2021)]{wang2021science}
Dashun Wang and Albert-L{\'a}szl{\'o} Barab{\'a}si.
\newblock \emph{The science of science}.
\newblock Cambridge University Press, 2021.

\bibitem[Wang et~al.(2023)Wang, Ma, Mao, Bai, Liang, and
  Li]{wang2023quantifying}
Shiyun Wang, Yaxue Ma, Jin Mao, Yun Bai, Zhentao Liang, and Gang Li.
\newblock Quantifying scientific breakthroughs by a novel disruption indicator
  based on knowledge entities.
\newblock \emph{Journal of the Association for Information Science and
  Technology}, 74\penalty0 (2):\penalty0 150--167, 2023.

\bibitem[Wu et~al.(2019)Wu, Wang, and Evans]{wu2019large}
Lingfei Wu, Dashun Wang, and James~A Evans.
\newblock Large teams develop and small teams disrupt science and technology.
\newblock \emph{Nature}, 566\penalty0 (7744):\penalty0 378--382, 2019.

\bibitem[Wu and Yan(2019)]{wu2019solo}
Qiang Wu and Zhaoyang Yan.
\newblock Solo citations, duet citations, and prelude citations: New measures
  of the disruption of academic papers.
\newblock \emph{arXiv preprint arXiv:1905.03461}, 2019.

\bibitem[Wu and Wu(2019)]{wu2019confusing}
Shijie Wu and Qiang Wu.
\newblock A confusing definition of disruption.
\newblock 2019.

\bibitem[Yang et~al.(2023)Yang, Deng, Wang, Zhang, and
  Yang]{yang2023disruptive}
Alex~J Yang, Sanhong Deng, Hao Wang, Yiqin Zhang, and Wenxia Yang.
\newblock Disruptive coefficient and 2-step disruptive coefficient: Novel
  measures for identifying vital nodes in complex networks.
\newblock \emph{Journal of Informetrics}, 17\penalty0 (3):\penalty0 101411,
  2023.

\bibitem[Zeng et~al.(2021)Zeng, Fan, Di, Wang, and Havlin]{zeng2021fresh}
An~Zeng, Ying Fan, Zengru Di, Yougui Wang, and Shlomo Havlin.
\newblock Fresh teams are associated with original and multidisciplinary
  research.
\newblock \emph{Nature human behaviour}, 5\penalty0 (10):\penalty0 1314--1322,
  2021.

\end{thebibliography}

\end{document}